\documentclass[letterpaper, 11pt]{llncs}

\usepackage[margin=1in]{geometry}
\usepackage[bookmarks, colorlinks=true, plainpages = false, citecolor = blue,linkcolor=red,urlcolor = blue, filecolor = blue]{hyperref}
\usepackage{url}
\usepackage{amsmath,amsfonts,amssymb,bm, verbatim,dsfont,mathtools}
\usepackage{algorithm,algorithmic,color,graphicx,appendix}
\usepackage{subfigure}
\usepackage{etoolbox}
\usepackage{tikz}
\usepackage{xr,xspace}
\usepackage{todonotes}
\usepackage{paralist}
\usepackage{caption}

\newtheorem{assumption}{Assumption}

\newcommand{\eproof}{\hfill $\Box$}

\usepackage{xspace,prettyref}
\newcommand{\diverge}{\to\infty}

\newcommand{\ones}{\mathbf 1}

\newcommand{\reals}{{\mathbb{R}}}

\newrefformat{eq}{(\ref{#1})}
\newrefformat{chap}{Chapter~\ref{#1}}
\newrefformat{sec}{Section~\ref{#1}}
\newrefformat{algo}{Algorithm~\ref{#1}}
\newrefformat{fig}{Fig.~\ref{#1}}
\newrefformat{tab}{Table~\ref{#1}}
\newrefformat{rmk}{Remark~\ref{#1}}
\newrefformat{clm}{Claim~\ref{#1}}
\newrefformat{def}{Definition~\ref{#1}}
\newrefformat{cor}{Corollary~\ref{#1}}
\newrefformat{lmm}{Lemma~\ref{#1}}
\newrefformat{prop}{Proposition~\ref{#1}}
\newrefformat{app}{Appendix~\ref{#1}}
\newrefformat{hyp}{Hypothesis~\ref{#1}}
\newrefformat{thm}{Theorem~\ref{#1}}

\newcommand{\pth}[1]{\left( #1 \right)}

\newcommand{\iprod}[2]{\left \langle #1, #2 \right\rangle}

\newcommand{\tx}{{\widetilde{x}}}

\newcommand{\calA}{{\mathcal{A}}}

\newcommand{\calC}{{\mathcal{C}}}

\newcommand{\calE}{{\mathcal{E}}}
\newcommand{\calF}{{\mathcal{F}}}

\newcommand{\calH}{{\mathcal{H}}}
\newcommand{\calI}{{\mathcal{I}}}

\newcommand{\calL}{{\mathcal{L}}}

\newcommand{\calN}{{\mathcal{N}}}

\newcommand{\calR}{{\mathcal{R}}}
\newcommand{\calS}{{\mathcal{S}}}

\newcommand{\calV}{{\mathcal{V}}}

\newcommand{\calX}{{\mathcal{X}}}

\newcommand{\argmin}{{\rm argmin}}
\newcommand{\argmax}{{\rm argmax}}

\begin{document}

\title{Fault-Tolerant Distributed Optimization (Part IV): \\
Constrained Optimization with Arbitrary Directed Networks
\thanks{This research is supported in part by National Science Foundation award NSF 1329681. 
Any opinions, findings, and conclusions or recommendations expressed here are those of the authors
and do not necessarily reflect the views of the funding agencies or the U.S. government.}}

\author{Lili Su \hspace*{1in} Nitin H. Vaidya}
\institute{Department of Electrical and Computer Engineering, and\\
Coordinated Science Laboratory\\
University of Illinois at Urbana-Champaign\\
Email:\{lilisu3, nhv\}@illinois.edu}
\maketitle

\begin{center}
Technical Report\\
~\\

\today
\end{center}
~

\centerline{\bf Abstract}

We study the problem of constrained distributed optimization in multi-agent networks when some of the computing agents may be faulty. In this problem, the system goal is to have all the non-faulty agents collectively minimize a global objective given by weighted average of local cost functions, each of which is initially known to a non-faulty agent only.
In particular, we are interested in the scenario when the computing agents are connected by an {\em arbitrary directed communication network}, some of the agents may suffer from {\em crash faults} or {\em Byzantine faults}, and the estimate of each agent is restricted to lie in a common constraint set.
This problem finds its applications in social computing and distributed large-scale machine learning.

It was shown in \cite{su2015byzantine} that it is {\em impossible} to optimize the {\em exact average} of the local functions at all the non-faulty agents. With this observation, the fault-tolerant multi-agent optimization problem
was first formulated in \cite{su2015byzantine} by introducing two problem parameters $\beta$ and $\gamma$, with $\beta$ as weight threshold and $\gamma$ as the minimum number of weights that exceed the threshold $\beta$ \cite{su2015byzantine}. The problem parameters $\beta$ and $\gamma$ together characterize the system performance.
We focus on the family of algorithms considered in \cite{su2015fault}, where only local communication and minimal memory carried across iterations are allowed.
In particular, we generalize our previous results on fully-connected networks and unconstrained optimization  \cite{su2015fault}  to arbitrary directed networks and constrained optimization.
%
%
%
%
%
%
%
As a byproduct, we provide a matrix representation for iterative approximate crash consensus. The matrix representation allows us to characterize the convergence rate for crash iterative consensus.


\section{Introduction}
\label{sec:intro}

There has been significant research on the problem of distributed optimization in multi-agent systems \cite{Duchi2012,nedic2015distributed,Nedic2009,Tsianos2012}. In a multi-agent system, each agent initially knows a convex cost function, and is capable of performing local computation as well as local communication, i.e., each agent can only exchange information with its neighbors.
 The system goal is to have all the agents collaboratively minimize the global objective given by the average of all the local cost functions. One application of this problem lies in social resource allocation, where each agent has his own cost function, and the system/society goal is to find a resource allocation solution such that the average of each of the agents' costs is minimized. Another application can be found in the domain of large scale distributed machine learning, where data are generated at different locations. The above multi-agent optimization problem is well-studied under the assumption that {\em every} computing agent is reliable throughout the execution \cite{Duchi2012,nedic2015distributed,Nedic2009,Tsianos2012}. However, as the complexity of multi-agent networks increases, it is becoming harder and harder to meet this assumption. In particular, due to the distributed fashion in data processing, some data may be lost during processing or be tampered by malicious local data managers. 
The need for robustness for distributed optimization problems has received some attentions recently \cite{Duchi2012}. 
In particular, Duchi et al.\ \cite{Duchi2012} studied the impact of random communication link failures on the convergence of distributed variant of dual averaging algorithm. Specifically, each realizable link failure pattern considered in \cite{Duchi2012} is assumed to admit a doubly-stochastic matrix which governs
the evolution dynamics of local estimates of the optimum.
However, we are not aware of
prior work that obtains the results presented in this report except our companion work \cite{su2015byzantine,su2015fault}, where both Byzantine faults and crash faults are considered, and the network is assumed to be fully-connected. 
 Agents suffering crash faults can unexpectedly stop participating in the prescribed protocol/algorithm, and agents suffering Byzantine faults \cite{Lynch:1996:DA:525656} can behave arbitrarily, and adversarially try to degrade the behavior of the system.  
In this report, we consider arbitrary directed communication networks. 
Similar algorithm structures and arbitrary directed networks are also considered in \cite{DBLP:journals/corr/SuV15a}, where the local functions are redundant in terms of the encoded information about input functions,
 and the goal is to have all the non-faulty agents collaboratively minimize the average of all the input functions.

Recently, Sundaram et al. \cite{Sundaram2015} looked at a similar problem
but with the faulty agents restricted to broadcasting their messages (sending identical messages) to their outgoing neighbors, and the global objective is simply a convex combination of local cost functions at the non-faulty agents. Their algorithm is shown to reach consensus on a value in the convex hull of the optima of the non-faulty functions. We show stronger guarantees.  
In addition, their results are based on the assumption that every update matrix has a common left eigenvector corresponding to eigenvalue 1.
In contrast, we consider the general Byzantine fault model, where the faulty agents can send different messages to different outgoing neighbors via point to point communication. Note that our fault model incorporates the fault model in \cite{Sundaram2015} as a special case. 
Among all the convex combination objectives, we measure the ``quality" of a global objective by two parameters ($\beta$ and $\gamma$, discussed later). Our focus is on design and analysis of the algorithms that are capable of solving high ``quality" objectives. In addition, our results do not rely on the common left eigenvector assumption.


\paragraph{\bf Contribution:} The main contributions of this report are two-fold. The first contribution is to generalize the results obtained in our previous work \cite{su2015fault} derived for fully-connected networks to arbitrary directed networks. The impact of constraints on the local estimates is also considered. The second contribution is to provide a matrix representation for the iterative approximate crash consensus, under which an explicit convergence rate is derived.  \\

Next, we formally restate the fault-tolerant multi-agent optimization problem, proposed in \cite{su2015byzantine}.
\subsection{ Problem Formulation}

The system under consideration is synchronous, and consists of $n$ agents connected by an {\em arbitrary directed} communication network $G(\calV,\calE)$, where $\calV=\{1,\dots,n\}$ is the set of $n$ agents, and $\calE$ is the set of directed edges between the agents in $\calV$. 

Up to $f$ of the $n$ agents may be faulty. Two fault models, crash faults and Byzantine faults, are considered respectively.
Agents suffering crash faults can unexpectedly stop participating in the prescribed protocol/algorithm, and agents suffering Byzantine faults \cite{Lynch:1996:DA:525656} can behave arbitrarily, and adversarially try to degrade the behavior of the system.
Let $\calF$ denote the set of faulty agents in a given execution and let $\calN=\calV-\calF$.
 The set $\calF$ of faulty agents may be chosen by an adversary arbitrarily.  Agent $i$ can reliably transmit messages to agent $j$ if and only if
the directed edge $(i,j)$ is in $\calE$.
Each agent can send messages to itself as well, however,
for convenience, we {\em exclude self-loops} from set $\calE$.
That is, $(i,i)\not\in\calE$ for $i\in\calV$.
With a slight abuse of terminology, we will use the terms {\em edge}
and {\em link} interchangeably, and use the terms {\em nodes}
and {\em agents} interchangeably in our presentation.

For each agent $i$, let $N_i^-$ be the set of agents from which $i$ has incoming
edges.
That is, $N_i^- = \{\, j ~|~ (j,i)\in \calE\, \}$.
Similarly, define $N_i^+$ as the set of agents to which agent $i$
has outgoing edges. That is, $N_i^+ = \{\, j ~|~ (i,j)\in \calE\, \}$.
Since we exclude self-loops from $\calE$,
$i\not\in N_i^-$ and $i\not\in N_i^+$.
However, we note again that each agent can indeed send messages to itself.
Agent $j$ is said to be an {\em incoming neighbor} of agent $i$,
if $j\in N_i^-$. Similarly, $j$ is said to be an {\em outgoing neighbor}
of agent $i$, if $j\in N_i^+$. In addition, define $d_i^{-}=|N_i^-|$ to be the incoming degree of agent $i$, and define $d_i^{+}=|N_i^+|$ to be the outgoing degree of agent $i$.

%
%
Let $\calX\subseteq \reals$ be a nonempty closed and convex set.
We say that a function $h: \calX\rightarrow \mathbb{R}$ is {\em admissible} if (i) $h(\cdot)$ is convex, continuously
differentiable, and has $L$-Lipschitz gradient, where $L$ is a fixed constant,
(ii) $h(\cdot)$ has bounded gradient, i.e., $|h^{\prime}(x)|\le L$ for each $x\in \calX$,
and (iii) the set $\arg\min_{x\in \calX} h(x)$ containing the optima of $h(\cdot)$
is non-empty and compact (i.e., bounded and closed). The bounded gradient assumption holds when $\calX$ is compact.
Each agent $i\in \calV$ is initially provided with an {\em admissible} local cost function $h_i: \calX\rightarrow\mathbb{R}$.

When $f=0$, i.e., every agent is guaranteed to be reliable, one commonly adopted global objective \cite{Duchi2012,nedic2015distributed,Nedic2009,Tsianos2012} is
\begin{align}
\label{objective f=0}
\frac{1}{n} \sum_{i=1}^n h_i(x),
\end{align}
i.e., the agents should collaboratively minimize the average cost of individual agents' costs. In distributed machine learning, this objective corresponds to the requirement that the data collected at different locations (local functions) be utilized equally in order to reduce bias. When $f>0$, the global objective (\ref{objective f=0}) cannot be minimized.
 This is because if an agent $i$ crashes, in particular crashes at the very beginning of an execution, then any information about the local function $h_i(x)$ is unavailable to the other agents. Also, since a Byzantine agent can lie arbitrarily about its local function, choosing the global objective to be (\ref{objective f=0}) may result in the system's output systematically biased by the Byzantine agents. Thus, under crash fault model, a proper global objective should be a {\em convex combination} of local functions -- in which a local function kept by a crashed agent may have weight 0.
  Under Byzantine fault model, a proper global objective should be a convex combination of {\em untampered functions} (the functions initially known by non-faulty agents) only. Two slightly different formulations, namely Problem 1 and Problem 2, initially proposed in \cite{su2015byzantine}, are formally described in (\ref{objective crash}) and (\ref{objective Byzantine}), respectively, with $\gamma$ and $\beta$ as problem parameters to characterize how good a proper global objective is.\\

\begin{align}
\label{objective crash}
{\bf Problem \, 1:}~ \text{Under crash}& \text{\, fault model, each non-faulty agent outputs}\nonumber \\
~~~~~~~~~~&\tx  \in \arg \min_{x\in \calX}\quad \sum_{i\in \calV} \alpha_i h_i(x)\\
\text{such that} & \nonumber \\
&\forall i\in\calV, ~ \alpha_i\geq 0,  \nonumber \\
&\sum_{i\in \calV}\alpha_i=1, \text{~~and~~} \sum_{i\in\calN} {\bf 1}(\alpha_i \ge \beta) ~ \geq ~ \gamma \nonumber
\end{align}
Problem 1 requires that the output $\tx$ be an optimum of a function formed as
a convex combination of local cost functions. More precisely, for some choice
of weights $\alpha_i$ for $i\in \calV$ such that $\alpha_i\geq 0$ and $\sum_{i\in \calV}\alpha_i=1$,
the output must be an optimum of the weighted cost function $\sum_{i\in \calV} \alpha_i\,h_i(x)$, with the coefficients $\alpha_i$'s representing the system's utilization level of untampered data (local functions). 
In addition, Problem 1 requires that a large enough number of untampered local functions (at least $\gamma$) be used nontrivially (i.e., with their weights lower bounded by $\beta$).
Note that
${\bf 1}\{\alpha_i \ge \beta\}$ is an indicator function that outputs 1 if $\alpha_i \ge\beta$,
and 0 otherwise.
%
Under Byzantine fault model, the problem is slightly different from the one under crash fault model. In particular, in Problem 2, the summation of the local functions is taken over non-faulty agents $\calN$ instead of taking over all the agents $\calV$ in (\ref{objective crash}). Both Problem 1 and Problem 2 require enough number of {\em non-faulty} functions have non-trivial weights.

\begin{align}
\label{objective Byzantine}
{\bf Problem \, 2:}~ \text{Under Byzantine}& \text{\, fault model, each non-faulty agent outputs}\nonumber \\
~~~&\tx  \in \arg \min_{x\in \calX}\quad \sum_{i\in \calN} \alpha_i h_i(x)\\
\text{such that} & \nonumber \\
&\forall i\in\calN, ~ \alpha_i\geq 0,  \nonumber \\
&\sum_{i\in \calN}\alpha_i=1, \text{~~and~~} \sum_{i\in\calN} {\bf 1}(\alpha_i \ge \beta) ~ \geq ~ \gamma \nonumber
\end{align}

Our problem formulations require that (in the time limit) all non-faulty agents output identical $\tx\in\mathbb{R}$, while
satisfying the constraints imposed by the problem (as listed in (\ref{objective crash})  and (\ref{objective Byzantine})).
Thus, the traditional crash consensus and Byzantine consensus \cite{impossible_proof_lynch} problems, which also impose a similar
{\em agreement} condition, are a special cases of our Problem 1 and Problem 2, respectively \cite{su2015byzantine}.

\subsection{Related Work}\label{sec: related work}

Fault-tolerant consensus \cite{PeaseShostakLamport} is a special case of the optimization problem considered in this report. There is a significant body of work on fault-tolerant consensus, including \cite{Chaudhuri92morechoices,Dolev:1986:RAA:5925.5931,fekete1990asymptotically,friedman2007asynchronous,mostefaoui2003conditions,vaidya2012iterative}.
The optimization algorithms presented in this report use fault-tolerant consensus as a component.

Convex optimization, including distributed convex optimization, also has a long history \cite{bertsekas1989parallel}. However, we are not aware of
prior work that obtains the results presented in this report except \cite{su2015byzantine,DBLP:journals/corr/SuV15a,su2015fault}.
Primal and dual decomposition methods that led themselves naturally to a distributed paradigm are well-known \cite{Boyd2011}. 
There has been significant research on a variant of distributed optimization problem \cite{Duchi2012,Nedic2009,Nedic2010,Tsianos2012}, in which the global objective $h(x)$ is a summation of $n$ convex functions, i.e, $h(x)=\sum_{j=1}^n h_j(x)$, with function $h_j(x)$ being known to the $j$-th agent. The need for robustness for distributed optimization problems has received some attentions recently \cite{Duchi2012,kailkhura2015consensus,marano2009distributed,su2015byzantine,DBLP:journals/corr/SuV15a,zhang2014distributed}. In particular, Duchi et al.\ \cite{Duchi2012} studied the impact of random communication link faults on the convergence of distributed variant of dual averaging algorithm. Specifically, each realizable link fault pattern considered in \cite{Duchi2012} is assumed to admit a doubly-stochastic matrix which governs
the evolution dynamics of local estimates of the optimum.

We considered Byzantine faults and crash faults in \cite{su2015byzantine,DBLP:journals/corr/SuV15a,su2015fault}. In particular,  \cite{su2015byzantine,DBLP:journals/corr/SuV15a} considered Byzantine faults under synchronous systems, and \cite{su2015fault} considered both Byzantine faults and crash faults under synchronous systems, with results partially generalizable to asynchronous systems. It is showed in \cite{su2015byzantine} under Byzantine faults that at most $|\calN|-f$ non-faulty functions can have non-zero weights. This observation led to the formulation of Problem 2 in  (\ref{objective Byzantine}). Six algorithms were proposed in \cite{su2015byzantine}. Algorithms with alternative structure, where only local communication is needed, is proposed in \cite{su2015fault} for crash faults and Byzantine faults, respectively.
We showed in \cite{DBLP:journals/corr/SuV15a} that when there are sufficient redundancy in the input functions (each input function is not exclusively kept by a single agent), it is possible to solve (\ref{objective f=0}), where the summation is taken over all input functions. In addition, a simple low-complexity iterative algorithm was proposed in \cite{DBLP:journals/corr/SuV15a}, and a tight topological condition for the existence of such iterative algorithms is identified.

Concurrently, Sundaram et al. \cite{Sundaram2015} looked at a similar problem with different focuses, where the faulty agents are restricted to broadcasting their messages (sending identical messages) to their outgoing neighbors, and the global objective is simply a convex combination of local cost functions at the {\em non-faulty} agents. Their algorithm performance is equivalent to simply running iterative Byzantine consensus on the local optima. In addition, their results are based on the assumption that every update matrix has a common left eigenvector corresponding to eigenvalue 1.
%

\subsection{Preliminaries}
Let $\calX\subseteq \reals$ be a nonempty set. Denote $Dist\pth{x, \calX}$ to be the standard Euclidean distance of $x$ from the set $\calX$.
\begin{align}
\label{dist}
Dist\pth{x, \calX}=\inf_{y\in \calX} |x-y|.
\end{align}
Henceforth, we assume that the set $\calX$ is nonempty, closed and convex.
We use $P_{\calX}[x]$ to denote the projection of the point $x$ on the set $\calX$, i.e.,
\begin{align*}
P_{\calX}[x]=\arg\min_{z\in \calX}\, \left | z-x\right |.
\end{align*}
We use the projection inequality and non-expansiveness properties,
i.e., for any $x$,
\begin{align}
&\pth{P_{\calX}[x]-x}\pth{y-P_{\calX}[x]}\ge 0~~~~\text{for all }y\in \calX.\label{proj1}\\
&\left | P_{\calX}[x]-P_{\calX}[y]\right |\le \left | x-y\right |~~~~\text{for all $x$ and $y$.}\label{nonexpansive}
\end{align}
We also use the properties given in the following lemma, previously proved by Nedic et al. in \cite{Nedic2010}.
\begin{lemma}\cite{Nedic2010}
\label{proj p}
Let $\calX$ be a nonempty closed convex set in $\reals$. For any $x\in \reals$, the following holds
\begin{align}
&(a)\quad \pth{P_{\calX}[x]-x}(x-y)\le -\left |P_{\calX}[x]-x \right |^2 ~~~~\text{for all $y\in \calX$},\\
&(b)\quad \left |P_{\calX}[x]-y \right |^2 \le \left | x-y\right |^2-\left | P_{\calX}[x]-x\right |^2~~~~\text{for all $y\in \calX$}.
\end{align}

\end{lemma}
%
%
%

\section{Byzantine Fault Tolerance}
As Byzantine fault model is more general than crash fault model, we study Byzantine failure first. The proof ideas in this section can be adapted to crash failures.
In this section, we analyze the performance of two algorithms, namely Algorithm 1 and Algorithm 2, which were proposed in \cite{DBLP:journals/corr/SuV15a} and \cite{su2015fault}, respectively. Note that Algorithm 1 is essentially a combination of the distributed optimization algorithm in \cite{Nedic2009} and the Byzantine consensus algorithm in \cite{vaidya2012IABC}.

We first briefly review the iterative approximate Byzantine consensus problem
, where only local communication, and minimal memory carried across iterations, are allowed.
\begin{definition}\cite{vaidya2012iterative}
\label{reduced Byzantine}
For a given graph $G(\calV, \calE)$, a reduced graph $\calH_b$ under Byzantine faults is a subgraph of $G(\calV, \calE)$ obtained by removing all the faulty agents from $\calV$ along with their edges; and (ii)
removing any additional up to $f$ incoming edges at each non-faulty agent.
\end{definition}
Let us denote the collection of all the reduced graphs for a given $G(\calV, \calE)$ by $R^b$. Thus, $\calN$ is the set of agents in each element in $R^b$. Let $\tau_b=|R^b|$. It is easy to see that
$\tau_b$ depends on $\calF$ as well as the underlying network $G(\calV, \calE)$, and it is finite. Let $\phi=\left | \calF\right |$. Thus $\phi\le f$.
\begin{definition}
A source component \footnote{The definition of a source is different from \cite{vaidya2012IABC}, wherein a source is defined as a strongly-connected component that cannot be reached by the outside nodes.} $S$ of a given graph $G(\calV, \calE)$ is the collection of agents each of which has a directed path to every other agent in $G(\calV, \calE)$.
\end{definition}
It can be easily checked that if a source component $S$ exists, it is a strongly-connected component in $G(\calV, \calE)$. In addition, a graph contains at most one source component.

\begin{theorem}\cite{vaidya2012IABC}
\label{Byzantine consensus}
Iterative approximate Byzantine consensus is solvable on $G(\calV, \calE)$ if and if every reduced graph (as per Definition \ref{reduced Byzantine}) of $G(\calV, \calE)$ has a source component.
\end{theorem}
If addition, it has been shown in \cite{vaidya2012IABC} that the source component in each reduced graph contains at least $f+1$ nodes.
Throughout this section, we assume that the underlying graph $G(\calV, \calE)$ satisfies the tight condition in Theorem \ref{Byzantine consensus}.
\begin{assumption}
\label{a1}
Every reduced graph of $G(\calV, \calE)$ under Byzantine faults contains a source component.
\end{assumption}

%
%
\begin{definition}
\label{validfByzantine}
Let $\calA(\beta, \gamma)$ be the collection of functions defined as follows:
\begin{align}
\nonumber
\calA(\beta, \gamma)~\triangleq ~\Bigg{\{}~p(x): p(x)&=\sum_{i\in \calN} \alpha_i h_i(x), ~\forall i\in\calN, ~ \alpha_i\geq 0,\\
&\sum_{i\in \calN}\alpha_i=1,\text{~~and~~}
\pth{\sum_{i\in\calN} {\bf 1}\left\{\alpha_i\ge \beta\right\}} ~ \geq ~ \gamma ~\Bigg{\}}
\label{Bvalid collection}
\end{align}

\end{definition}
Each function in $\calA(\beta, \gamma)$ is called a valid function for a given tuple $(\beta, \gamma)$.
Define
\begin{align}
\label{union opt set cons}
X(\beta, \gamma)~\triangleq ~ \cup_{p(x)\in \calA(\beta, \gamma)} \, \arg\min_{x\in \reals} \, p(x).
\end{align}
The next lemma characterizes the properties of set $X(\beta, \gamma)$. Note that for each valid function $p(\cdot)\in \calA(\beta, \gamma)$, the minimization in (\ref{union opt set cons}) is taken over the whole real line $\reals$, instead of the constraint set $\calX$.
\begin{lemma}
\label{convex B1}
If $\beta\le \frac{1}{|\calN|}$ and $\gamma\le |\calN|$, the set $X(\beta, \gamma)$ is convex and closed.
\end{lemma}
The proof of Lemma \ref{convex B1} is similar to the proof of Lemma 10 and Lemma 11 in \cite{su2015fault}.

Note that $X(\beta, \gamma)$ is the collection of the unconstrained optimal solutions of all the valid functions in $\calA(\beta, \gamma)$. However, it is possible that there exists a point in $X(\beta, \gamma)$ that is infeasible (i.e., outside $\calX$). Let $Y(\beta, \gamma)$ be the collection of constrained optimal solutions for functions in $\calA(\beta, \gamma)$, formally defined as
\begin{align}
\label{union opt set}
Y(\beta, \gamma)~\triangleq ~ \cup_{p(x)\in \calA(\beta, \gamma)} \, \arg\min_{x\in \calX} \, p(x).
\end{align}

\begin{lemma}
\label{convex B cons}
If $\beta\le \frac{1}{|\calN|}$ and $\gamma\le |\calN|$, the set $Y(\beta, \gamma)$ is convex and closed.
\end{lemma}
\begin{proof}
We consider three cases: (1) $X(\beta, \gamma)\subseteq \calX$, (2) $X(\beta, \gamma)\cap \calX ~=~\O$ and (3) $X(\beta, \gamma)\cap \calX ~\not=~\O$ and $X(\beta, \gamma)\not\subseteq \calX$.

\paragraph{Case 1 ($X(\beta, \gamma)\subseteq \calX$):}

Since $X(\beta, \gamma)\subseteq \calX$, every point in $X(\beta, \gamma)$ is feasible. It is easy to see that $Y(\beta, \gamma)=X(\beta, \gamma)$. By Lemma \ref{convex B1}, we know that $Y(\beta, \gamma)$ is convex and closed.
\paragraph{Case 2 ($X(\beta, \gamma)\cap \calX ~=~\O$):}
Recall that both $X(\beta, \gamma)$ and $\calX$ are convex and closed.
Since $X(\beta, \gamma)\cap \calX ~=~\O$, either $\max \, X(\beta, \gamma)< \min \, \calX$, or $\min \, X(\beta, \gamma)>  \max \, \calX$ is true.
By symmetry, without loss of generality, assume that
$$\max \, X(\beta, \gamma)~< ~ \min \, \calX.$$
For ease of notation, let $x_0=\min \calX$ and $x_1=\max \calX$.
Let $p(\cdot)$ be an arbitrary function in $\calA(\beta, \gamma)$.
Since $\arg\min_{x\in \reals}\, p(x)\subseteq \, X(\beta, \gamma)$, and $p^{\prime}(\cdot)$ is non-decreasing, it holds that $p^{\prime}(x)>0$ for each $x\in \calX$. Thus,
\begin{align}
\label{boundary proj}
\arg\min_{x\in \calX}\, p(x)=\{x_0\}.
\end{align}
Since (\ref{boundary proj}) holds for any $p(\cdot)\in \calA(\beta, \gamma)$, we have
\begin{align*}
X(\beta, \gamma) ~&=~\cup_{p(\cdot)\in \calA(\beta, \gamma)}\, \arg\min_{x\in \calX}\, p(x)\\
~&=~\cup_{p(\cdot)\in \calA(\beta, \gamma)}\, \{x_0\}\\
~&=~\{x_0\},
\end{align*}
which is trivially convex and closed.

\paragraph{Case 3 ($X(\beta, \gamma)\cap \calX ~\not=~\O$ and $X(\beta, \gamma)\not\subseteq \calX$):}

%
Let $p(\cdot)$ be an arbitrary valid function in $\calA(\beta, \gamma)$. Either $\arg\min_{x\in \reals} \, p(x)\cap \calX \not=\O$ or $\arg\min_{x\in \reals} \, p(x)\cap \calX =\O$ is true.

Suppose $\arg\min_{x\in \reals} \, p(x)\cap \calX \not=\O$.
Let $y\in \arg\min_{x\in \reals} \, p(x)\cap \calX $, then $p^{\prime}(y)=0$ and $y\in \calX.$
Thus, $y\in \arg\min_{x\in \calX}\, p(x)$, which implies that
\begin{align}
\label{sub 1}
\arg\min_{x\in \reals} \, p(x)\cap \calX \subseteq \, \arg\min_{x\in \calX} \, p(x).
\end{align}
In addition, since $\arg\min_{x\in \reals} \, p(x)\cap \calX\not=\O$, then $p^{\prime}(y)=0$ for each $y\in \arg\min_{x\in \calX}\, p(x)$.
Otherwise, the optimality of $y$ will be violated.
Thus, $y\in \arg\min_{x\in \reals} \, p(x)\cap \calX$ and
\begin{align}
\label{sub 2}
\arg\min_{x\in \calX} \, p(x) \subseteq \, \arg\min_{x\in \reals} \, p(x)\cap \calX.
\end{align}
(\ref{sub 1}) and (\ref{sub 2}) together show that when $\arg\min_{x\in \reals} \, p(x)\cap \calX\not=\O$,
\begin{align}
\label{sub 3}
\arg\min_{x\in \calX} \, p(x) = \, \arg\min_{x\in \reals} \, p(x)\cap \calX.
\end{align}

Now consider the case when $\arg\min_{x\in \reals} \, p(x)\cap \calX =\O$.
Since $\arg\min_{x\in \reals} p(x) \cap \calX=\O$, either $\max \pth{\arg\min_{x\in \reals} \, p(x)} < x_0$ or $\min \pth{\arg\min_{x\in \reals} \, p(x)} > x_1$.
If
\begin{align}
\label{b1}
\max \pth{\arg\min_{x\in \reals} \, p(x)} < x_0,
\end{align}
by the analysis in case 2, we know that
$ \arg\min_{x\in \calX} \, p(x)=\{x_0\}.$
Similarly, if
\begin{align}
\label{b2}
\min \pth{\arg\min_{x\in \reals} \, p(x)} > x_1,
\end{align}
we have
$ \arg\min_{x\in \calX} \, p(x)=\{x_1\}.$\\

Let $\calI_1$, $\calI_2$, and $\calI_3$ be a partition of set $\calA(\beta, \gamma)$ such that
\begin{align*}
\calI_1&=\{p(\cdot): ~ p(\cdot)\in \calA(\beta, \gamma) ~~\text{and}~~(\ref{b1})~~\text{holds} \}\\
\calI_2&=\{p(\cdot): ~ p(\cdot)\in \calA(\beta, \gamma) ~~\text{and}~~(\ref{b2})~~\text{holds} \}\\
\calI_3&=\{p(\cdot): ~ p(\cdot)\in \calA(\beta, \gamma) ~~\text{and}~~\arg\min_{x\in \reals}\, p(x)\, \cap \calX\not=\O~ \}.
\end{align*}
It is easy to see that $\calI_1$, $\calI_2$, and $\calI_3$ together form a partition of set $\calA(\beta, \gamma)$.\\

Next, we show that
\begin{align}
\label{b11}
\pth{\cup_{p(x)\in \calI_1} \arg \min_{x\in \calX} \, p(x)} \subseteq \, X(\beta, \gamma)\cap \calX.
\end{align}

When $\calI_1=\O$, it holds that $\cup_{p(x)\in \calI_1} \arg \min_{x\in \calX} \, p(x)=\O$. Thus, (\ref{b11}) follows trivially.

When $\calI_1\not=\O$, then $\cup_{p(x)\in \calI_1} \arg \min_{x\in \calX} \, p(x)=\{x_0\}$, and $\min \pth{X(\beta, \gamma)\cap \calX}=x_0$. Otherwise, $\cup_{p(x)\in \calI_1} \arg \min_{x\in \calX} \, p(x)=\{x_0\}$ could not be true, because $X(\beta, \gamma)\cap \calX ~\not=~\O$ in case 3.
Thus,
$$\pth{\cup_{p(x)\in \calI_1} \arg \min_{x\in \calX} \, p(x)}=\{x_0\}\subseteq  X(\beta, \gamma)\cap \calX,$$
proving (\ref{b11}).

Similarly, we can show that
\begin{align}
\label{b12}
\pth{\cup_{p(x)\in \calI_2} \arg \min_{x\in \calX}\, p(x) } \subseteq \, X(\beta, \gamma)\cap \calX.
\end{align}

We get,
\begin{align}
\nonumber
&Y(\beta, \gamma)=\cup_{p(\cdot)\in \calA(\beta, \gamma)} \arg\min_{x\in \calX} \, p(x)\\
\nonumber
&=\left [  \cup_{p(\cdot)\in \calI_1} \arg\min_{x\in \calX} \, p(x) \right] \cup \left [  \cup_{p(\cdot)\in \calI_2} \arg\min_{x\in \calX} \, p(x) \right]\cup \left [  \cup_{p(\cdot)\in \calI_3} \arg\min_{x\in \calX} \, p(x) \right]\\
\nonumber
&=\left [  \cup_{p(\cdot)\in \calI_1} \arg\min_{x\in \calX} \, p(x) \right] \cup \left [  \cup_{p(\cdot)\in \calI_2} \arg\min_{x\in \calX} \, p(x) \right]\cup \left [  \cup_{p(\cdot)\in \calI_3} \pth{\arg\min_{x\in \reals} \, p(x)\cap \calX} \right]~~\text{by (\ref{sub 3})}\\
&=\left [  \cup_{p(\cdot)\in \calI_1} \arg\min_{x\in \calX} \, p(x) \right] \cup \left [  \cup_{p(\cdot)\in \calI_2} \arg\min_{x\in \calX} \, p(x) \right]\cup \left [  \cup_{p(\cdot)\in \calI_1\cup \calI_2\cup\calI_3} \pth{\arg\min_{x\in \reals} \, p(x)\cap \calX} \right] \label{b21}\\
\nonumber
&=\left [  \cup_{p(\cdot)\in \calI_1} \arg\min_{x\in \calX} \, p(x) \right] \cup \left [  \cup_{p(\cdot)\in \calI_2} \arg\min_{x\in \calX} \, p(x) \right]\cup \left [ \pth{\cup_{p(\cdot)\in \calA(\beta,\gamma)} \arg\min_{x\in\reals} p(x) }\cap \calX\right ]\\
\nonumber
&=\left [  \cup_{p(\cdot)\in \calI_1} \arg\min_{x\in \calX} \, p(x) \right] \cup \left [  \cup_{p(\cdot)\in \calI_2} \arg\min_{x\in \calX} \, p(x) \right]\cup \left [ X(\beta, \gamma)\cap \calX\right ]\\
\nonumber
&=X(\beta, \gamma)\cap \calX ~~\text{by (\ref{b11}) and (\ref{b12})}
\end{align}
Equality (\ref{b21}) holds because $\arg\min_{x\in \reals} \, p(x)\cap \calX=\O,$ for each $p(\cdot)\in \calI_1\cup \calI_2$.

Since $Y(\beta, \gamma)=X(\beta, \gamma)\cap \calX $, and both $X(\beta, \gamma)$ and $\calX$ are convex and closed, it holds that
$Y(\beta, \gamma)$ is also convex and closed. \\

Case 1, case 2 and case 3 together prove Lemma \ref{convex B cons}.

\eproof
\end{proof}

\subsection{Exchange Local Estimates}

Let $\{\lambda[t]\}_{t=0}^{\infty}$ be a sequence of stepsizes such that $\lambda[t]\le \lambda[t+1]$ for all $t\ge 0$, $\sum_{t=0}^{\infty} \lambda[t]=\infty$, and $\sum_{t=0}^{\infty} \lambda^2[t]<\infty$. Let $x_i[0]$ be the initial state of agent $i\in \calV$.

\paragraph{}
\vspace*{8pt}\hrule

~
{\bf Algorithm 1 } for agent $i$ for iteration $t\ge 1$
\vspace*{4pt}\hrule

~

\begin{enumerate}

\item {\em Transmit step:} Transmit current state $x_i[t-1]$ on all outgoing edges.
\item {\em Receive step:} Receive values on all incoming edges. These values form
multiset\footnote{In a multiset, multiple instances of of an element is allowed.  For instance, $\{1, 1, 2\}$ is a multiset. } $r_i[t]$ of size $d_i^-=|N_i^{-}|$.


\item {\em Update step:}
Sort the values in $r_i[t]$ in an increasing order, and eliminate
the smallest $f$ values, and the largest $f$ values (breaking ties
arbitrarily)\footnote{Note that if $G(\calV, \calE)$ satisfies Assumption \ref{a1}, then $d_i\ge 2f+1$ for each $i\in \calV$}.
 Let $N_i^*[t]$ denote the identifiers of agents from
whom the remaining $|N_i^{-}| - 2f$ values were received, and let
$w_j$ denote the value received from agent $j\in N_i^*[t]$.
For convenience, define $w_i[t-1]=x_i[t-1]$. \footnote{Observe that
if $j\in \{i\}\cup N_i^*[t]$ is non-faulty, then $w_j=x_j[t-1]$.}

Update its state as follows.
\begin{eqnarray}
x_i[t] ~ = ~P_{\calX}\left [\frac{1}{d_i^-+1-2f}\pth{\sum_{j\in \{i\}\cup N_i^*[t]}  \, w_j[t-1]}-\lambda[t-1]~ h_i^{\prime}(x_i[t-1])\right ],
\label{e_Z}
\end{eqnarray}
where $h_i^{\prime}(x_i[t-1])$ is the gradient of agent $i$'s local function $h_i(\cdot)$ at $x_i[t-1]$.
\end{enumerate}

~
\hrule

~

~

Note that $x_i[t]\in \calX$, for each $i\in \calN$ and each $t\ge 1$. Define $v_i[t-1]$ and ${\epsilon}_i[t-1]$ as follows.
\begin{align}
v_i[t-1]&=\frac{1}{d_i^-+1-2f}\pth{\sum_{j\in \{i\}\cup N_i^*[t]}  \, w_j[t-1]},\label{proj11}\\
e_i[t-1]&=P_{\calX}\left [v_i[t-1]-\lambda[t-1]~ h_i^{\prime}(x_i[t-1])\right ]-\pth{v_i[t-1]-\lambda[t-1]~ h_i^{\prime}(x_i[t-1])} \label{proj12}.
\end{align}
Then the update of $x_i[t]$ in (\ref{e_Z}) can be rewritten as
\begin{align}
x_i[t]=v_i[t-1]-\lambda[t-1]~ h_i^{\prime}(x_i[t-1])+e_i[t-1].
\end{align}

The following proposition, proved in \cite{Nedic2010}, states that the projection error diminishes over time.
We present the proof here for completeness.
\begin{proposition}\cite{Nedic2010}
\label{proj e}
For each $i\in \calN$ and each $t\ge 0$, the projection error $e_i[t-1]$ satisfies
$$ \left|e_i[t]\right|\le \lambda[t]\,L.$$
\end{proposition}
\begin{proof}
\begin{align*}
\left|x_i[t+1]-v_i[t]\right |^2&= \left|P_{\calX}\left [v_i[t]-\lambda[t]~ h_i^{\prime}(x_i[t])\right ]-v_i[t]\right |^2\\
&\le \left|v_i[t]-\lambda[t]~ h_i^{\prime}(x_i[t])-v_i[t]\right |^2-\left|e_i[t]\right |^2~~~\text{by Lemma \ref{proj p}}\\
&=\lambda^2[t]\left|h_i^{\prime}(x_i[t])\right |^2-\left|e_i[t]\right |^2\\
&\le \lambda^2[t]L^2-\left|e_i[t]\right |^2~~~\text{since $|h_i^{\prime}(x)|\le L$ for any $x\in \calX$}.
\end{align*}
Thus,
\begin{align*}
\left|e_i[t]\right |^2\le \lambda^2[t]L^2-\left|x_i[t+1]-v_i[t]\right |^2 \le \lambda^2[t]L^2.
\end{align*}

\eproof
\end{proof}

Let $d_{\max}=\max_{j\in \calV} d_j^-$. Now we proceed to analyze the performance of Algorithm 1 in terms of $\beta$ and $\gamma$.

\begin{theorem}
\label{BT}
For a given graph $G(\calV, \calE)$ and $\beta \le  \frac{1}{(2(d_{\max}+1-2f))^{\tau_b(n-\phi)}}$, if each reduced graph $\calH_b$ contains a source component with size at least $\gamma$, where $\gamma\ge f+1$, then Algorithm 1 optimizes a function in $\calA(\beta, \gamma)$.
%
%
\end{theorem}

The proof of Theorem \ref{BT} relies on several lemmas and theorems proved in our previous work \cite{DBLP:journals/corr/SuV15a,su2015fault}. Next, we simply state them when needed without giving a proof.\\

 Without loss of generality, let us assume that the non-faulty agents are indexed as 1 to $n-\phi$.
Recall that the system is synchronous.
If a non-faulty agent does not receive an expected message from an incoming neighbor (in the {\em Receive step} below), then that message is assumed to have some default value. 

Recall that
 $i\not\in N_i^*[t]$
because $(i,i)\not\in\calE$.
The ``weight'' of each term on the right-hand side of
(\ref{e_Z}) is $\frac{1}{d_i^{-}+1-2f}$, and these weights add to 1. Observe that $0<\frac{1}{d_i^{-}+1-2f}\leq 1$.
Let ${\bf x}[t]\in \reals^{n- \phi}$ be a real vector of dimension $n-\phi$, with $x_i[t]$ being the local estimate of agent $i, \forall\,  i\in \calN$ at the end of iteration $t$, let ${\bf d}[t]\in \reals^{n- \phi}$ be a real vector with $d_i[t]$ being the gradient of function $h_i(\cdot)$ at $x_i[t]$, and let ${\bf e}[t]\in \reals^{n-\phi}$ be a real vector of projection errors defined in (\ref{proj12}).

Since the source component exists in every reduced graph of $G(\calV,\calE)$, and is of size at least $\gamma \ge f+1$, 
\cite{Vaidya2012MatrixConsensus} implies that 
%
%
\begin{align}
\label{update}
{\bf x}[t+1]={\bf M}[t]{\bf x}[t]- \lambda[t] {\bf d}[t]+{\bf e}[t].
\end{align}
The construction of ${\bf M}[t]$ and relevant properties are given in \cite{Vaidya2012MatrixConsensus}. 
Let $\calH\in R^b$ be a reduced graph of the given graph $G(\calV, \calE)$ with ${\bf H}$ as adjacency matrix. It is shown in \cite{Vaidya2012MatrixConsensus} that every iteration $t$, and for every ${\bf M}[t]$, there exists a reduced graph $\calH [t]\in R^b$ with adjacency matrix ${\bf H}[t]$ such that
\begin{align}
\label{reducedgraph}
{\bf M}[t]\ge \xi\, {\bf H}[t],
\end{align}
where $\xi=\frac{1}{2\pth{d_{\max}+1-2f}}$. It is easy to see that $\frac{1}{2n} \le \xi <1$.  Equation (\ref{update}) can be further expanded out as
\begin{align}
\nonumber
{\bf x}[t+1]&={\bf M}[t]{\bf x}[t]-\lambda[t]{\bf d}[t]+{\bf e}[t]\\
\nonumber
&={\bf M}[t]\pth{{\bf M}[t-1]{\bf x}[t-1]-\lambda[t-1]{\bf d}[t-1]+{\bf e}[t-1]}-\lambda[t]{\bf d}[t]+{\bf e}[t]\\
\nonumber
&\cdots\\
\nonumber
&=\pth{{\bf M}[t]{\bf M}[t-1]\cdots {\bf M}[0]}{\bf x}[0]-\lambda[t-1]\sum_{r=1}^{t+1}\pth{{\bf M}[t]{\bf M}[t-1]\cdots {\bf M}[r]}{\bf d}[r-1]\\
\nonumber
&\quad+\sum_{r=1}^{t+1}\pth{{\bf M}[t]{\bf M}[t-1]\cdots {\bf M}[r]}{\bf e}[r-1]\\
&={\bf\Phi}(t,0){\bf x}[0]-\lambda[t-1]\sum_{r=1}^{t+1}{\bf \Phi}(t, r){\bf d}[r-1]+\sum_{r=1}^{t+1}{\bf\Phi}(t,r){\bf e}[r-1],
\label{updates}
\end{align}
where ${\bf \Phi}(t,r)={\bf M}[t]{\bf M}[t-1]\ldots {\bf M}[r]$ and by convention ${\bf \Phi}(t,t)={\bf M}[t]$ and ${\bf \Phi} (t, t+1)={\bf I}_{n-\phi}$, the identity matrix. Note that ${\bf \Phi}(t,r)$ is a backward product (i.e., the index decreases from left to right in the product).

\subsubsection{Convergence of the Transition Matrices ${\bf \Phi}(t,r)$ }
\label{ConvergenceProduct}

It can be seen from (\ref{updates}) that the evolution of estimates of non-faulty agents ${\bf x}[t]$ is determined by the backward product ${\bf \Phi}(t,r)$. Thus, we first characterize the evolutional properties and limiting behaviors of the backward product ${\bf \Phi}(t,r)$.

Recall that $\tau_b=|R^b|$ is the total number of reduced graphs (under Byzantine faults) of the given $G(\calV,\calE)$, $\phi=|\calF|\le f$ is the actual number of faulty agents in a given execution, and $\xi\in (\frac{1}{2n}, 1)$ in (\ref{reducedgraph}). Let $\nu=\tau_b(n-\phi)$ and $\theta=1-\xi^{\nu}$. 
The following lemma describes the structural property of ${\bf \Phi}(t,r)$ for sufficient large $t$. For a given $r$, Lemma \ref{lb} states that all non-faulty agents will be influenced by at least $\gamma$ common non-faulty agents, and this set of influencing agents may depend on $r$. 
\begin{lemma}\cite{DBLP:journals/corr/SuV15a}
\label{lb}
There are at least $\gamma$ columns in ${\bf \Phi}(r+\nu-1,r)$ that are lower bounded by $\xi^{\nu}\ones$ component-wise for all $r$, where $\ones\in \reals^{n-\phi}$ is an all one column vector of dimension $n-\phi$.
\end{lemma}
Using coefficients of ergodicity theorem, it is showed in \cite{Vaidya2012MatrixConsensus} that ${\bf \Phi}(t,r)$ is weak-ergodic. Moreover, because weak-ergodicity is equivalent to strong-ergodicity for backward product of stochastic matrices \cite{1977Seneta}, as $t\diverge$ the limit of ${\bf \Phi}(t,r)$ exists
\begin{align}
\label{mixing}
\lim_{t\ge r,~ t\diverge}{\bf \Phi}(t, r)=\ones {\bf \pi^{\prime}}(r),
\end{align}
where ${\bf \pi^{\prime}}(r)\in \reals^{n-\phi}$ is the transpose of vector ${\bf \pi}(r)$, which is a stochastic vector (may depend on $r$).  

\begin{theorem}\cite{Anthonisse1977360}
\label{convergencerate}
Let $\nu=\tau(n-\phi)$ and $\theta=1-\xi^{\nu}$. For any sequence ${\bf \Phi}(t, r)$,
\begin{align}
\left | {\bf \Phi}_{ij}(t, r)-\pi_j(r)\right |\le \theta^{\lceil\frac{t-r+1}{\nu}\rceil},
\end{align}
for all $t\ge r$.
\end{theorem}

Our next lemma is an immediate consequence of Lemma \ref{lb} and the convergence of ${\bf \Phi}(t, r)$, stated in (\ref{mixing}).
\begin{lemma}\cite{DBLP:journals/corr/SuV15a}
\label{lblimiting}
For any fixed $r$, at least $\gamma$ entries in $\pi (r)$ are lower bounded by $\xi^{\nu}$, i.e., there exists a subset $\calI_r\subseteq \calN$ such that $|\calI_r|\ge \gamma$ and
\begin{align*}
\pi_i(r)\ge \xi^{\nu},
\end{align*}
for each $i\in \calI_r$.
\end{lemma}

\subsubsection{Convergence Analysis of Algorithm 1}
Here, we study the convergence behavior of Algorithm 1. The structure of our convergence proof is rather standard, which is also adopted in \cite{Duchi2012,Nedic2009,ram2010distributed,DBLP:journals/corr/SuV15a,su2015fault,Tsianos2012,tsitsiklis1986distributed}. We have shown that the evolution dynamics of ${\bf x}[t]$ is captured by (\ref{update}) and (\ref{updates}).
Suppose that all agents, both non-faulty agents and faulty agents cease computing $h_i^{\prime}(x_i[t])$ after some time $\bar{t}$, i.e., after $\bar{t}$ gradient is replaced by 0.

Let $\{\bar{\bf x}[t]\}$ be the sequences of local estimates generated by the non-faulty agents in this case. From (\ref{updates}) we get
\begin{align*}
\bar{\bf x}[t]={\bf x}[t],
\end{align*}
for all $t\le \bar{t}$. From (\ref{update}) and (\ref{updates}), we have for all $s\ge 0$, it holds that
\begin{align}
\bar{\bf x}(\bar{t}+s+1)&={\bf \Phi} (\bar{t}+s, 0){\bf x}[0]-\sum_{r=1}^{\bar{t}}\lambda[r-1]{\bf \Phi} (\bar{t}+s, r){\bf d}[r-1]+\sum_{r=1}^{\bar{t}}{\bf \Phi} (\bar{t}+s, r){\bf e}[r-1]
\label{evo}\\
\nonumber
&={\bf \Phi} (\bar{t}+s, \bar{t})\pth{{\bf \Phi} (\bar{t}-1, 0){\bf x}[0]-\sum_{r=1}^{\bar{t}}\lambda[r-1]{\bf \Phi} (\bar{t}-1, r){\bf d}[r-1]+\sum_{r=1}^{\bar{t}}{\bf \Phi} (\bar{t}-1, r){\bf e}[r-1]}\\
&={\bf \Phi} (\bar{t}+s, \bar{t}){\bf x}[\bar{t}]~~~\text{by (\ref{updates})~~~~~~}\label{update at t}
\end{align}
Note that the summation in RHS of (\ref{evo}) is over $\bar{t}$ terms since all agents cease computing $h_j^{\prime}(x_j[t])$ starting from iteration $\bar{t}$. As $s\diverge$, we have

\begin{align}
\lim_{s\diverge}
\bar{\bf x}(\bar{t}+s+1)&=\lim_{s\diverge}\pth{{\bf \Phi} (\bar{t}+s, 0){\bf x}[0]-\sum_{r=1}^{\bar{t}}\lambda[r-1]{\bf \Phi} (\bar{t}+s, r){\bf d}[r-1]+\sum_{r=1}^{\bar{t}}{\bf \Phi} (\bar{t}+s, r){\bf e}[r-1]}
\nonumber\\
&=\lim_{s\diverge}{\bf \Phi} (\bar{t}+s, 0){\bf x}[0]-\sum_{r=1}^{\bar{t}}\lambda[r-1]\lim_{s\diverge}{\bf \Phi} (\bar{t}+s, r){\bf d}[r-1]+\sum_{r=1}^{\bar{t}}\lim_{s\diverge}{\bf \Phi} (\bar{t}+s, r){\bf e}[r-1]\nonumber\\
&=\ones {\bf \pi^{\prime}}(0){\bf x}[0]-\sum_{r=1}^{\bar{t}}\lambda[r-1]\ones {\bf \pi^{\prime}}(r){\bf d}[r-1]+ \sum_{r=1}^{\bar{t}}\ones {\bf \pi^{\prime}}(r){\bf e}[r-1]\nonumber\\
&= \pth{\iprod{\pi(0)}{{\bf x}[0]}-\sum_{r=1}^{\bar{t}}\lambda[r-1] \iprod{\pi(r)}{{\bf d}[r-1]}+\sum_{r=1}^{\bar{t}}\iprod{{\bf \pi}(r)}{{\bf e}[r-1]}} \ones,
\label{identical}
\end{align}
where $\iprod{\cdot}{\cdot}$ is used to denote the inner product of two vectors of proper dimension. Let ${\bf y}[\bar{t}]$ denote the limiting vector of $\bar{\bf x}(\bar{t}+s+1)$ as $s+1\diverge$. Since all entries in the limiting vector are identical we denote the identical value by $y[\bar{t}]$. Thus, ${\bf y}[\bar{t}]=[y[\bar{t}], \ldots, y[\bar{t}]]^{\prime}$.

From (\ref{identical}) we have
\begin{align}
y[\bar{t}]=\iprod{\pi(0)}{{\bf x}[0]}-\sum_{r=1}^{\bar{t}}\lambda[r-1] \iprod{\pi(r)}{{\bf d}[r-1]}+\sum_{r=1}^{\bar{t}}\iprod{{\bf \pi}(r)}{{\bf e}[r-1]}.
\label{yupdate}
\end{align}
In addition, by (\ref{update at t}), an alternative expression of $y[\bar{t}]$ is obtained.
\begin{align}
\label{convex limit y}
y[\bar{t}]=\iprod{\pi(\bar{t})}{{\bf x}[\bar{t}]}=\sum_{j=1}^{n-\phi} \pi_{j}(\bar{t})x_j[\bar{t}].
\end{align}
If, instead, all agents cease computing $h_i^{\prime}(x_i[t])$ after iteration $\bar{t}+1$,  then the identical value, denoted by $y[(\bar{t}+1)]$, similar to (\ref{yupdate}), equals
\begin{align}
\nonumber
y[(\bar{t}+1)]&=\iprod{\pi(0)}{{\bf x}[0]}-\sum_{r=1}^{\bar{t}+1}\lambda[r-1] \iprod{\pi(r)}{{\bf d}[r-1]}+\sum_{r=1}^{\bar{t}+1}\iprod{{\bf \pi}(r)}{{\bf e}[r-1]}\\
\nonumber
&=\iprod{\pi(0)}{{\bf x}[0]}-\sum_{r=1}^{\bar{t}}\lambda[r-1] \iprod{\pi(r)}{{\bf d}[r-1]}+\sum_{r=1}^{\bar{t}}\iprod{{\bf \pi}(r)}{{\bf e}[r-1]}\\
\nonumber
&\quad-\lambda[\bar{t}]  \iprod{\pi[(\bar{t}+1)]}{{\bf d}[\bar{t}]}+\iprod{\pi[(\bar{t}+1)]}{{\bf e}[\bar{t}]}\\
&=y[\bar{t}]-\lambda[\bar{t}]  \iprod{\pi[(\bar{t}+1)]}{{\bf d}[\bar{t}]}+\iprod{\pi[(\bar{t}+1)]}{{\bf e}[\bar{t}]},
\label{ydynamic}
\end{align}
where $d_i[\bar{t}]=h_i^{\prime}(x_i[t])$, for each $i\in \calN$. 
With a little abuse of notation, henceforth we use $t$ to replace $\bar{t}$. The actual reference of $t$ should be clear from the context.\\

It was shown in \cite{DBLP:journals/corr/SuV15a} that the difference $|y[t]-x_i[t]|$ shrinks over time. 
We prove a similar claim here.
\begin{lemma}
\label{consensus}
Let $\{x_i[t]\}$ be the iterates generated by Algorithm 1 and consider the auxiliary sequence defined in (\ref{yupdate}). If $\lim_{t\diverge} \lambda[t]~=~0$, then
$$\lim_{t\diverge} \left | x_i[t]-y[t]\right |~=~0.$$
\end{lemma}
\begin{proof}
Recall (\ref{updates}). For $t>0$,
\begin{align*}
{\bf x}[t+1]={\bf\Phi}(t,0){\bf x}[0]-\lambda[t-1]\sum_{r=1}^{t+1}{\bf \Phi}(t, r){\bf d}[r-1]+\sum_{r=1}^{t+1}{\bf\Phi}(t,r){\bf e}[r-1]
\end{align*}
then each $x_i(t)$ can be written as
\begin{align*}
x_i[t+1]&=\sum_{j=1}^{n-\phi}{\bf \Phi}_{ij} (t, 0)x_{j}[0]-\sum_{r=1}^{t+1}\pth{\lambda[r-1]\sum_{j=1}^{n-\phi}{\bf \Phi}_{ij} (t, r)h_{j}^{\prime}(x_j[r-1])}\\
&\quad+\sum_{r=1}^{t+1}\sum_{j=1}^{n-\phi}{\bf\Phi}_{ij}(t,r)e_j[r-1];
\end{align*}
and (\ref{yupdate}) implies that
\begin{align*}
y[t+1]&=\sum_{j=1}^{n-\phi} \pi_j(0)x_j[0]-\sum_{r=1}^{t+1}\lambda[r-1] \sum_{j=1}^{n-\phi}\pi_j(r)\,h_j^{\prime}\pth{x_j[r-1]}+\sum_{r=1}^{t+1} \sum_{j=1}^{n-\phi}\pi_j(r)\, e_j[r-1]
\end{align*}

 Thus

\begin{align}
\label{uniformbd}
\nonumber
&|y[t+1]-x_i[t+1]|\\
\nonumber
&=\Bigg{|}\sum_{j=1}^{n-\phi} \pi_j(0)x_j[0]-\sum_{r=1}^{t+1}\lambda[r-1] \sum_{j=1}^{n-\phi}\pi_j(r)\,h_j^{\prime}\pth{x_j[r-1]}+\sum_{r=1}^{t+1} \sum_{j=1}^{n-\phi}\pi_j(r)\, e_j[r-1]\\
\nonumber
&\quad -\sum_{j=1}^{n-\phi}{\bf \Phi}_{ij} (t, 0)x_{j}[0]+\sum_{r=1}^{t+1}\pth{\lambda[r-1]\sum_{j=1}^{n-\phi}{\bf \Phi}_{ij} (t, r)h_{j}^{\prime}(x_j[r-1])}-\sum_{r=1}^{t+1}\sum_{j=1}^{n-\phi}{\bf\Phi}_{ij}(t,r)e_j[r-1]\Bigg{|}\\
\nonumber
&\le \left |\sum_{j=1}^{n-\phi}\pth{\pi_{j} (0)-{\bf \Phi}_{ij} (t, 0)}x_{j}(0)\right |+\left |\sum_{r=1}^{t+1}\pth{\lambda[r-1]\sum_{j=1}^{n-\phi}\pth{{\bf \Phi}_{ij} (t, r)-\pi_{j}(r)}h_{j}^{\prime}(x_j[r-1])}\right |\\
&\quad + \left | \sum_{r=1}^{t+1} \sum_{j=1}^{n-\phi}\pth{\pi_j(r)-{\bf\Phi}_{ij}(t,r)}\, e_j[r-1]\right |.
\end{align}
We bound the three terms in (\ref{uniformbd}) separately. 
%
%
%
For the first term in (\ref{uniformbd}), we have
\begin{align}
\label{RHS1}
\nonumber
\left |\sum_{j=1}^{n-\phi}\pth{\pi_{j} (0)-{\bf \Phi}_{ij} (t, 0)}x_{j}[0]\right |&\le \sum_{j=1}^{n-\phi}\left |\pi_{j} (0)-{\bf \Phi}_{ij} (t, 0)\right |\, |x_{j}[0]|\\
\nonumber
&\overset{(a)}{\le} \sum_{j=1}^{n-\phi}\theta^{\lceil \frac{t+1}{\nu}\rceil}\max \{|u|, |U|\}\\
&=\pth{n-\phi}\max \{|u|, |U|\}\theta^{\lceil \frac{t+1}{\nu}\rceil},
\end{align}
where inequality (a) follows from Theorem \ref{convergencerate}.

The second term in (\ref{uniformbd}) can be bounded as follows.
\begin{align}
\nonumber
&\left |\sum_{r=1}^{t+1}\pth{\lambda[r-1]\sum_{j=1}^{n-\phi}\pth{{\bf \Phi}_{ij} (t, r)-\pi_{j}(r)}h_{j}^{\prime}(x_j[r-1])}\right |\\
\nonumber
&\quad  \overset{(a)}{\le}  \sum_{r=1}^{t+1}\pth{\lambda[r-1]\sum_{j=1}^{n-\phi}\left |{\bf \Phi}_{ij} (t, r)-\pi_{j}(r)\right |\,  \left |h_{j}^{\prime}(x_j[r-1])\right |}+\lambda[t] \left |h_{i}^{\prime}(x_i[t])-\sum_{j=1}^{n-\phi}\pi_j(t)h_{j}^{\prime}(x_j[t])\right |\\
\nonumber
&\quad  \le \sum_{r=1}^{t+1}\pth{\lambda[r-1]\sum_{j=1}^{n-\phi}\left |{\bf \Phi}_{ij} (t, r)-\pi_{j}(r)\right |\, \left |h_{j}^{\prime}(x_j[r-1])\right |}+\lambda[t]\sum_{j=1}^{n-\phi}\pi_j(t)\left |h_{i}^{\prime}(x_i[t])-h_{j}^{\prime}(x_j[t])\right |\\
\nonumber
&\quad  \le \sum_{r=1}^{t+1}\pth{\lambda[r-1]\sum_{j=1}^{n-\phi}\left |{\bf \Phi}_{ij} (t, r)-\pi_{j}(r)\right |}L+2\lambda[t] L\\
&\quad \le \pth{n-\phi}L\sum_{r=1}^{t+1}\lambda[r-1] \theta^{\lceil \frac{t-r+1}{\nu}\rceil}+2\lambda[t] L~~~\text{by Theorem \ref{convergencerate}}
\label{RHS2}
\end{align}
where inequality $(a)$ follows from the fact that ${\bf \Phi}(t-1,t)={\bf I}$. Note that when $t=1$, it holds that
$$\sum_{r=1}^{t-1}\pth{\lambda[r-1]\sum_{j=1}^{n-\phi}|{\bf \Phi}_{ij} (t-1, r)-\pi_{j}(r)|\, |h_{j}^{\prime}(x_j[r-1])|}=0.$$

In addition, the third term in (\ref{uniformbd}) can be bounded as follows.
\begin{align}
\label{RHS3}
\nonumber
\left | \sum_{r=1}^{t+1} \sum_{j=1}^{n-\phi}\pth{\pi_j(r)-{\bf\Phi}_{ij}(t,r)}\, e_j[r-1]\right |&\le \sum_{r=1}^{t+1} \sum_{j=1}^{n-\phi}\left | \pi_j(r)-{\bf\Phi}_{ij}(t,r)\right | \left | e_j[r-1]\right |\\
\nonumber
&\le \sum_{r=1}^{t+1} \sum_{j=1}^{n-\phi}\left | \pi_j(r)-{\bf\Phi}_{ij}(t,r)\right | \lambda[r-1]L~~~\text{by Proposition \ref{proj e}}\\
&\le (n-\phi) L \sum_{r=1}^{t+1} \lambda[r-1]\theta^{\lceil \frac{t+1-r}{\nu}\rceil}~~~\text{by Theorem \ref{convergencerate}}
\end{align}

From (\ref{RHS1}) and (\ref{RHS2}), the LHS of (\ref{uniformbd}) can be upper bounded by
\begin{align*}
|y[t+1]-x_i[t+1]|\le \pth{n-\phi}\max \{|u|, |U|\}\theta^{\lceil \frac{t+1}{\nu}\rceil} +2\pth{n-\phi}L\sum_{r=1}^{t+1}\lambda[r-1] \theta^{\lceil \frac{t-r+1}{\nu}\rceil}+2\lambda[t] L.
\end{align*}

It is shown in \cite{DBLP:journals/corr/SuV15a} that
\begin{align*}
\lim_{t\diverge} ~\sum_{r=1}^{t+1}\lambda[r-1] \theta^{\lceil \frac{t-r+1}{\nu}\rceil}~=~0.
\end{align*}
Then, taking limit over $t$, we have
\begin{align*}
\lim_{t\diverge}\,|y[t+1]-x_i[t+1]|&\le~ \lim_{t\diverge}\,\pth{n-\phi}\max \{|u|, |U|\}\theta^{\lceil \frac{t+1}{\nu}\rceil} \\
&\quad+2\pth{n-\phi}L\lim_{t\diverge}\,\sum_{r=1}^{t+1}\lambda[r-1] \theta^{\lceil \frac{t-r+1}{\nu}\rceil}+\lim_{t\diverge}\,2\lambda[t] L\\
&=~0
\end{align*}

\raggedleft $\square$
\end{proof}

Recall that $\gamma$ (as per Theorem \ref{BT}) is the minimal size of the source component of each reduced graph of $G(\calV, \calE)$, and that $Y(\beta, \gamma)$, defined in (\ref{union opt set}), is the union of the optimal sets of all valid functions in $\calA(\beta, \gamma)$.
 Next we show that when $\beta \le  \frac{1}{\pth{2(d_{\max}\, +\, 1-2f)}^{\tau_b(n-\phi)}}$, for each $i\in \calN$, the following holds
\begin{align}
\label{convergece1}
\lim_{t\diverge} Dist\pth{x_i[t], Y(\beta, \gamma)} = 0.
\end{align}

\vskip 2\baselineskip
Since $Dist\pth{\cdot, Y(\beta, \gamma)}$ is a metric, by (\ref{dist}) and the triangle inequality, we get
\begin{align}
\label{tri}
Dist\pth{x_i[t], Y(\beta, \gamma)}\le |x_i[t] - y[t] |+Dist\pth{y[t], Y(\beta, \gamma)}.
\end{align}
By (\ref{tri}) and Lemma \ref{consensus}, we know that to show (\ref{convergece1}) holds, it is enough to show that
\begin{align}
\label{convergece2}
\lim_{t\diverge} Dist\pth{y[t], Y(\beta, \gamma)}~=~0.
\end{align}

We first informally describe the intuition of why (\ref{convergece2}) should hold. We will make present a rigorous argument later.

Recall that ${\bf d}[t]$ is a real vector with $d_i[t]$ being the gradient of function $h_i(\cdot)$ at $x_i[t]$. By Lemma \ref{consensus} and Proposition \ref{proj e}, we know that
$x_i[t]$ is asymptotic to $y[t]$ for each $i\in \calN$, and $e_i[t]$ is diminishing, i.e., for each $i\in \calN$,
$$ y[t]\approx x_i[t],$$
and $$e_i[t]\approx 0,$$
for sufficient large $t$. Consequently, informally speaking, for sufficient large $t$, the update of $y[t]$ in (\ref{ydynamic}) roughly equals
\begin{align}
\label{intution y}
\nonumber
y(t+1)&=y[t]-\lambda[t]  \iprod{\pi(t+1)}{{\bf d}[t]}+\iprod{\pi[(t+1)]}{{\bf e}[t]}\\
\nonumber
&= y[t]-\lambda[t]\sum_{j=1}^{n-\phi} \pi_j(t+1) h_j^{\prime}(x_j[t])+\sum_{j=1}^{n-\phi}\pi_j(t+1)e_j[t]\\
&\approx y[t]-\lambda[t]\sum_{j=1}^{n-\phi} \pi_j(t+1) h_j^{\prime}(y[t]).
\end{align}
In addition, for each $t\ge 0$, define function $p_{t+1}(\cdot)$ as follows
\begin{align}
\label{al1 valid}
p_{t+1}(x)=\sum_{j=1}^{n-\phi} \pi_j(t+1) h_j(x),
\end{align}
for each $x\in \reals$.
By Lemma \ref{lblimiting}, we know that there exists $\calI_{t+1}\subseteq \calN$ such that $|\calI_{t+1}|\ge \gamma$ and $\pi_j(t+1)\ge \xi^{\nu}$ for each $j\in \calI_{t+1}$. Thus,
\begin{align*}
p_{t+1}(\cdot)\in \calA(\beta, \gamma),
\end{align*}
where $\beta \le\xi^{\nu}=\frac{1}{(2(d_{\max}\,+\,1-2f))^{\nu}}=\frac{1}{(2(d_{\max}\,+\,1-2f))^{\tau_b(n-\phi)}}$, and $\gamma$ is the minimal size of the source component in each reduced graph of $G(\calV, \calE)$. That is, the function $p_{t+1}(\cdot)$, defined in (\ref{al1 valid}), is a valid function.

By (\ref{intution y}), informally speaking, we know that for sufficient large $t$,  $y[t]$ comes closer to $\argmin_{x\in \reals}\, p_{t+1}(x)$ at each iteration $t$. Observing that $\argmin_{x\in \reals}\, p_{t+1}(x)\subseteq Y(\beta, \gamma)$ and $Y(\beta, \gamma)\subseteq \reals$ is convex, we know that at each iteration $t$, $y[t]$ comes closer to $Y(\beta, \gamma)$. Consequently,
\begin{align*}
\lim_{t\diverge} Dist\pth{y[t], Y(\beta, \gamma)}~=~0,
\end{align*}
proving (\ref{convergece2}).

\vskip 2\baselineskip

Next, we present a rigorous analysis validating the above intuition.

Recall from (\ref{ydynamic}) that
\begin{align*}
y(t+1)&=y[t]-\lambda[t]\sum_{j=1}^{n-\phi} \pi_j(t+1) h_j^{\prime}(x_j[t])+\sum_{j=1}^{n-\phi}\pi_j(t+1)e_j[t]\\
&=y[t]-\lambda[t]\sum_{j=1}^{n-\phi} \pi_j(t+1) h_j^{\prime}(y[t])+\sum_{j=1}^{n-\phi}\pi_j(t+1)e_j[t]\\
&\quad+\lambda[t]\sum_{j=1}^{n-\phi} \pi_j(t+1) h_j^{\prime}(y[t])-\lambda[t]\sum_{j=1}^{n-\phi} \pi_j(t+1) h_j^{\prime}(x_j[t])\\
&=y[t]-\lambda[t]\sum_{j=1}^{n-\phi} \pi_j(t+1) h_j^{\prime}(y[t]) +\sum_{j=1}^{n-\phi}\pi_j(t+1)e_j[t]+\lambda[t]\sum_{j=1}^{n-\phi} \pi_j(t+1) \pth{h_j^{\prime}(y[t])- h_j^{\prime}(x_j[t])}\\
&=y[t]-\lambda[t]p_{t+1}^{\prime}(y[t]) +\sum_{j=1}^{n-\phi}\pi_j(t+1)e_j[t]+
\lambda[t]\sum_{j=1}^{n-\phi} \pi_j(t+1) \pth{h_j^{\prime}(y[t])- h_j^{\prime}(x_j[t])}~~~\text{by (\ref{al1 valid})}
\end{align*}
Then,
\begin{align}
\label{port 1}
\nonumber
&Dist\pth{y[t+1], Y(\beta, \gamma)}\\
\nonumber
&=Dist\pth{y[t]-\lambda[t]p_{t+1}^{\prime}(y[t])+\sum_{j=1}^{n-\phi}\pi_j(t+1)e_j[t] +
\lambda[t]\sum_{j=1}^{n-\phi} \pi_j(t+1) \pth{h_j^{\prime}(y[t])- h_j^{\prime}(x_j[t])}, Y(\beta, \gamma)}\\
\nonumber
&=\inf_{z\in Y(\beta, \gamma)} \left |y[t]-\lambda[t]p_{t+1}^{\prime}(y[t]) +\sum_{j=1}^{n-\phi}\pi_j(t+1)e_j[t]+
\lambda[t]\sum_{j=1}^{n-\phi} \pi_j(t+1) \pth{h_j^{\prime}(y[t])- h_j^{\prime}(x_j[t])}-z \right |\\
\nonumber
&\le \inf_{z\in Y(\beta, \gamma)} \left |y[t]-\lambda[t]p_{t+1}^{\prime}(y[t])-z \right |+ \left | \sum_{j=1}^{n-\phi}\pi_j(t+1)e_j[t] +\lambda[t]\sum_{j=1}^{n-\phi} \pi_j(t+1) \pth{h_j^{\prime}(y[t])- h_j^{\prime}(x_j[t])}\right |\\
\nonumber
&\le \inf_{z\in Y(\beta, \gamma)} \left |y[t]-\lambda[t]p_{t+1}^{\prime}(y[t])-z \right |+
\sum_{j=1}^{n-\phi} \pi_j(t+1)\left |e_j[t]\right |
+\lambda[t]\sum_{j=1}^{n-\phi} \pi_j(t+1) \left |h_j^{\prime}(y[t])- h_j^{\prime}(x_j[t])\right |\\
&\le \inf_{z\in Y(\beta, \gamma)} \left |y[t]-\lambda[t]p_{t+1}^{\prime}(y[t])-z \right |+\lambda[t] L+ \lambda[t]\sum_{j=1}^{n-\phi} \pi_j(t+1) L \left |y[t]- x_j[t]\right |,
\end{align}
where the last inequality follows from Proposition \ref{proj e} and the fact that $h_j^{\prime}(\cdot)$ is $L$--Lipschitz continuous for each $j\in \calN$.

Let $M[t]=\max_{i\in \calN} x_i[t]$ and $m[t]=\min_{i\in \calN} x_i[t]$.
Recall from (\ref{convex limit y}) that
\begin{align*}
y[t]=\iprod{\pi(t)}{{\bf x}[t]}=\sum_{i=1}^{n-\phi} \pi_{i}(t)x_i[t].
\end{align*}
Then
\begin{align*}
\left |y[t]- x_j[t]\right |=\left |\sum_{i=1}^{n-\phi} \pi_{i}(t)x_i[t]- x_j[t]\right |\le M[t]-m[t].
\end{align*}
In addition, since $\sum_{j=1}^{n-\phi} \pi_j(t+1)=1$, it holds that
\begin{align}
\label{port 11}
\sum_{j=1}^{n-\phi} \pi_j(t+1) L \left |y[t]- x_j[t]\right |\le \sum_{j=1}^{n-\phi} \pi_j(t+1) L \pth{M[t]-m[t]}=L\pth{M[t]-m[t]}.
\end{align}
By (\ref{port 11}), we simplify (\ref{port 1}) as follows.
\begin{align}
\label{port 111}
Dist\pth{y[t+1], Y(\beta, \gamma)} \le \inf_{z\in Y(\beta, \gamma)} \left |y[t]-\lambda[t]p_{t+1}^{\prime}(y[t])-z \right |+\lambda[t] L+ \lambda[t]L\pth{M[t]-m[t]},
\end{align}
which is similar to equation (68) in \cite{su2015fault}, where we replace
$x_{j_{t+1}^{\prime}}$ by $y[t]$, replace $$\frac{1}{\left |\calR^{j_{t+1}}[t] \right |}\sum_{k\in \calR^{j_{t+1}}[t] } h_k^{\prime}(x_{j_{t+1}^{\prime}}),$$
by $p_{t+1}^{\prime}(y[t])$, and replace $y$ by $z$. Since $\lim_{t\diverge}\, \lambda[t]L~=~0 $, the remaining proof is similar to the proof of Theorem 2 in \cite{su2015fault}.

\subsection{Exchange both Local Estimates and Gradients}
\label{sec: algorithm byzantine}
Recall that $x_i[0]$ is the initial state of agent $i\in \calV$.
\paragraph{}
\vspace*{8pt}\hrule
~

{\bf Algorithm 2} for agent $j$ for iteration $t\ge 1$:
~
\vspace*{4pt}\hrule

\begin{list}{}{}
\item[{\bf Step 1:}]
Compute $h_j^{\prime}\pth{x_j[t-1]}$ -- the gradient of the local cost function $h_j(\cdot)$ at point $x_j[t-1]$, and send the estimate and gradient pair $(x_j[t-1], h_j^{\prime}\pth{x_j[t-1]})$ to on all outgoing edges. \\

~
\item[{\bf Step 2:}]
Let $\calR_j[t-1]$ denote the multi-set of tuples of the form $\pth{x_i[t-1], \, h_i^{\prime}(x_i[t-1])}$ received on all incoming edges 
as a result of step 1.\\

In step 2, agent $j$ should be able to receive a tuple $(w_i[t-1], g_i[t-1])$ from each agent $i\in  N_j^-$. For non-faulty agent $i\in N_j^-\cap \calN$, $w_i[t-1]=x_i[t-1]$ and $g_i[t-1]=h_i^{\prime}\pth{x_i[t-1]}$. If a faulty agent $k\in N_j^-\cap\calF$ does not send a tuple to agent $j$, then agent $j$ assumes $(w_k[t-1], g_k[t-1])$ to be some default tuple. \footnote{In contrast to Algorithms 1, 2 and 3 in \cite{su2015byzantine}, the adopted default tuple in Algorithm 2 here is not necessarily known to all agents. In addition, the default tuple may vary across iterations. }\\

~

\item[{\bf Step 3:}] Sort the first entries of the received tuples in $\calR_j[t-1]$ in a non-increasing order (breaking ties arbitrarily), and erase the smallest $f$ values and the largest $f$ values \footnote{Note that if $G(\calV, \calE)$ satisfies Assumption \ref{a1}, then $d_i\ge 2f+1$ for each $i\in \calV$}. Let $\calR_j^{1}[t-1]$ be the identifiers of the $|N_j^-|-2f=d_j^--2f$ agents from whom the remaining first entries were received. Similarly, sort the second entries of the received tuples in $\calR_j[t-1]$ together with the gradient $h_j^{\prime}(x_j[t-1])$ of agent $j$ in a non-increasing order (breaking ties arbitrarily), and erase the smallest $f$ values and the largest $f$ values. Let $\calR_j^{2}[t-1]$ be the identifiers of the $d_j^--2f+1$ agents from whom the remaining second entries were received.
Denote the largest and smallest gradients among the remaining values by $\hat{g}_j[t-1]$ and $\check{g}_j[t-1]$, respectively. Set $\widetilde{g}_j[t-1]=\frac{1}{2}\pth{\hat{g}_j[t-1]+\check{g}_j[t-1]}$.

Update its state as follows.
\begin{align}
\label{Byzantine Iterative}
x_j[t]=P_{\calX}\left [\frac{1}{d_j^-+1-2f} \pth{\sum_{i\in \calR_j^1[t-1]\cup \{j\}}w_i[t-1]}-\lambda[t-1]\widetilde{g}_j[t-1]\right].
\end{align}

\end{list}

\hrule

~

%
%
%
For each $i\in \calN$ and each $t\ge 1$, define $v_i[t-1]$ and ${\epsilon}_i[t-1]$ as follows.
\begin{align}
v_i[t-1]&=\frac{1}{d_j^-+1-2f} \pth{\sum_{i\in \calR_j^1[t-1]\cup \{j\}}w_i[t-1]},\label{proj21}\\
e_i[t-1]&=P_{\calX}\left [v_i[t-1]-\lambda[t-1]\widetilde{g}_j[t-1]\right ]-\pth{v_i[t-1]-\lambda[t-1]\widetilde{g}_j[t-1]} \label{proj22}.
\end{align}
Then the update of $x_i[t]$ in (\ref{e_Z}) can be rewritten as
\begin{align}
x_i[t]=v_i[t-1]-\lambda[t-1]\widetilde{g}_j[t-1]+e_i[t-1].
\end{align}

Note that in (\ref{Byzantine Iterative}), the averaging strategy with respect to the received estimates is different from that with respect to the received gradients. The averaging involving the received estimates is widely adopted \cite{vaidya2012IABC,ren2007information}. As we will see later that averaging remained extremes admits a desired representation of (\ref{Byzantine Iterative}), resulting in better algorithm performance in terms of $\gamma$.

Similar to the previous subsection, without loss of generality, assume agents indexed from 1 through $n-\phi$ are non-faulty, and agents indexed from $n-\phi+1$ to $n$ are faulty.
Let ${\bf x}[t-1]\in \reals^{n-\phi}$ be a real vector of the local estimates at the beginning of iteration $t$ with ${\bf x}_j[t-1]=x_j[t-1]$ being the local estimate of agent $j\in \calN$, let $\widetilde{\bf g}[t-1]\in  \reals^{n-\phi}$ be a vector of the local gradients at iteration $t$ with $\widetilde{\bf g}_j[t-1]=\widetilde{ g}_j[t-1], j\in \calN$, and ${\bf e}[t-1]\in \reals^{n-\phi}$ be the vector of projection error defined in (\ref{proj22}).
Recall that the graph $G(\calV, \calE)$ satisfies Assumption \ref{a1}.
As shown in \cite{Vaidya2012MatrixConsensus}, the update of ${\bf x}\in \reals^{n-\phi}$ in each iteration can be written compactly in a matrix form.
\begin{align}
\label{matrix representation}
{\bf x}[t]={\bf M}[t-1]{\bf x}[t-1]-\lambda[t-1]\widetilde{\bf g}[t-1]+{\bf e}[t-1].
\end{align}

The construction of ${\bf M}[t]$ is the same as the construction in the previous subsection. 

Equation (\ref{matrix representation}) can be further expanded out as
\begin{align}
\label{MR evo BS}
{\bf x}[t]={\bf  \Phi}(t-1, 0){\bf x}[0]-\sum_{r=0}^{t-1} \lambda[r]{\bf \Phi}(t-1, r+1)\widetilde{\bf g}[r]+\sum_{r=0}^{t-1} {\bf \Phi}(t-1, r+1){\bf e}[r].
\end{align}

When the source component of every reduced graph of $G(\calV, \calE)$ exists, then asymptotic consensus can be achieved. 

\begin{lemma}
\label{BS consensus}
For $i,j\in \calN$,
$$\lim_{t\diverge}|x_i[t]-x_j[t]|=0.$$
\end{lemma}
The proof of Lemma \ref{BS consensus} is similar to the proof of Corollary 3 in \cite{su2015fault}.
%
%
%
%
%
%
%
%
%
%
%
The following proposition is first observed in Proposition 2 in \cite{su2015fault}.
\begin{proposition}\cite{su2015fault}
\label{p1}
Let $a, b, c, d\in \reals$ such that
$b<a, b\le c\le \frac{1}{2}\pth{a+b}, \frac{1}{2}\pth{a+b}<a \le d.$ Then there exists $0\le \xi\le 1$, for which $\frac{1}{2}\pth{a+b}=\xi c+(1-\xi)d$ holds, and
$$\frac{1}{2}\le \xi \le 1.$$
\end{proposition}

Recall that $|\calF|=\phi$. For a given set $\calF$ and for each $i\in \calN$, let $|N_i^{-}\cap \calF|=\phi_i$.
Next we prove a key lemma.
\begin{lemma}
\label{BS valid gradient}
Let $\tilde{\beta}=\min\{\frac{1}{2\max_{i\in \calN} (d_i^-+1-\phi_i-f) }, \, \frac{1}{|\calN|}\}$, and $\tilde{\gamma}=\min_{i\in \calN} (d_i^-+1-\phi_i-f)$.
For each non-faulty agent $j\in \calN$ and each iteration $t\ge 1$, there exists a valid function $p(x)=\sum_{i\in \calN} \alpha_i\, h_i(x)\in \calA(\tilde{\beta}, \tilde{\gamma})$ such that
$$\widetilde{g}_j[t-1]=\sum_{i\in \calN} \alpha_i\, h_i^{\prime}(x_i[t-1]).$$
\end{lemma}

\begin{proof}
Recall that $\calR_j^2[t-1]$ denotes the set of agents from whom the remaining $d_j+1-2f$ gradient values were received in iteration $t$, and let us denote by $\calL_j[t-1]$ and $\calS_j[t-1]$ the set of agents from whom the largest $f$ gradient values and the smallest $f$ gradient values were received in iteration $t$. 

Let $i^*, j^*\in \calR_j^2[t-1]$ such that $g_{i^*}[t-1]=\check{g}_j[t-1]$ and $g_{j^*}[t-1]=\hat{g}_j[t-1]$.
Recall that $|N_j^{-}\cap \calF|=\phi_j$. Let $\calL_j^*[t-1]\subseteq \calL_j[t-1]-\calF$ and $\calS_j^*[t-1]\subseteq \calS_j[t-1]-\calF$ such that
$$|\calL_j^*[t-1]|=f-\phi_j+|\calR_j^2[t-1]\cap \calF|,$$
and
$$|\calS_j^*[t-1]|=f-\phi_j+|\calR_j^2[t-1]\cap \calF|.$$

We consider two cases: (i) $\hat{g}_j[t-1]>\check{g}_j[t-1]$ and (ii) $\hat{g}_j[t-1]=\check{g}_j[t-1]$, separately.

\paragraph{{\bf Case (i)}: $\hat{g}_j[t-1]>\check{g}_j[t-1]$.}

By definition of $\calL_j^*[t-1]$ and $\calS_j^*[t-1]$, we have
\begin{align}
\label{case 1 valid gradient}
\frac{1}{f-\phi_j+|\calR_j^2[t-1]\cap \calF|}\sum_{i\in \calS_j^*[t-1]}g_i[t-1]\le \widetilde{g}_j[t-1]\le \frac{1}{f-\phi_j+|\calR_j^2[t-1]\cap \calF|}\sum_{i\in \calL_j^*[t-1]}g_i[t-1].
\end{align}
Thus, there exists $0\le \xi\le 1$ such that
\begin{align}
\label{extrem nonfaulty}
\nonumber
\widetilde{g}_j[t-1]&=\xi\pth{\frac{1}{f-\phi_j+|\calR_j^2[t-1]\cap \calF|}\sum_{i\in \calS_j^*[t-1]}g_i[t-1]}\\
\nonumber
&\quad+(1-\xi)\pth{\frac{1}{f-\phi_j+|\calR_j^2[t-1]\cap \calF|}\sum_{i\in \calL_j^*[t-1]}g_i[t-1]}\\
&=\frac{\xi}{f-\phi_j+|\calR_j^2[t-1]\cap \calF|}\sum_{i\in \calS_j^*[t-1]}g_i[t-1]+\frac{1-\xi}{f-\phi_j+|\calR_j^2[t-1]\cap \calF|}\sum_{i\in \calL_j^*[t-1]}g_i[t-1].
\end{align}
Let $k\in \calR_j^2[t-1]-\calF$. By symmetry, assume $\xi\ge \frac{1}{2}$ and $\check{g}[t-1]\le g_k[t-1]\le \widetilde{g}_j[t-1]$. Since $|\calL_j[t-1]\cup \{j^*\}|=f+1$, there exists a non-faulty agent $j^{\prime}_k\in \calL_j[t-1]\cup \{j^*\}$. Thus,
$$g_{j^{\prime}_k}[t-1]\ge \hat{g}_j[t-1]>\widetilde{g}_j[t-1]\ge g_k[t-1]\ge \check{g}[t-1],$$
 and there exists $0\le \xi_k\le 1$ such that
\begin{align}
\label{BS middle nonfaulty}
\frac{1}{2}\pth{\hat{g}_j[t-1]+\check{g}_j[t-1]}=\widetilde{g}_j[t-1]=\xi_k g_k[t-1]+(1-\xi_k) g_{j^{\prime}_k}[t-1].
\end{align}
Let $a=\hat{g}_j[t-1], b=\check{g}_j[t-1], c=g_k[t-1],$ and $d=g_{j^{\prime}_k}[t-1]$. By Proposition \ref{p1}, we know that
$\frac{1}{2}\le \xi_k\le 1$.
%
%
Since
\begin{align*}
|(N_j^-\cup \{j\})\cap \calN|-f&=d_j^-+1-\phi_j-f=d_j^-+1-2f +f-\phi_j\\
&=\left|\calR_j^2[t-1]\right|+f-\phi_j=\left|\calR_j^2[t-1]-\calF\right|+\left|\calR_j^2[t-1]\cap\calF\right|+f-\phi_j,
\end{align*}
we get
\begin{align*}
\nonumber
\widetilde{g}_j[t-1]&=\frac{|(N_j^-\cup \{j\})\cap \calN|-f}{|(N_j^-\cup \{j\})\cap \calN|-f}\widetilde{g}_j[t-1]\\
&=\frac{|\calR_j^2[t-1]-\calF|}{|(N_j^-\cup \{j\})\cap \calN|-f}\widetilde{g}_j[t-1]+\frac{f-\phi_j+|\calR_j^2[t-1]\cap \calF|}{|(N_j^-\cup \{j\})\cap \calN|-f}\widetilde{g}_j[t-1]\\
\nonumber
&=\frac{1}{|(N_j^-\cup \{j\})\cap \calN|-f}\pth{\sum_{k\in \calR_j^2[t-1]-\calF}\widetilde{g}_j[t-1]}+\frac{f-\phi_j+|\calR_j^2[t-1]\cap \calF|}{|(N_j^-\cup \{j\})\cap \calN|-f}\widetilde{g}_j[t-1]\\
\nonumber
&=\frac{1}{|(N_j^-\cup \{j\})\cap \calN|-f}\sum_{k\in \calR_j^2[t-1]-\calF}\pth{\xi_k g_k[t-1]+(1-\xi_k) g_{j^{\prime}_k}[t-1]}]~~~\text{by}~(\ref{BS middle nonfaulty})\\
\nonumber
&\quad +\frac{\xi}{|(N_j^-\cup \{j\})\cap \calN|-f}\sum_{i\in \calS_j^*[t-1]}g_i[t-1]+\frac{1-\xi}{|(N_j^-\cup \{j\})\cap \calN|-f}\sum_{i\in \calL_j^*[t-1]}g_i[t-1]~~~\text{by}~(\ref{extrem nonfaulty})\\
\nonumber
&=\frac{1}{|(N_j^-\cup \{j\})\cap \calN|-f}\sum_{k\in \calR_j^2[t-1]-\calF}\pth{\xi_k\, h_k^{\prime}(x_k[t-1])+(1-\xi_k)\, h_{j^{\prime}_k}^{\prime}(x_{j^{\prime}_k}[t-1])}\\
&\quad +\frac{\xi}{|(N_j^-\cup \{j\})\cap \calN|-f}\sum_{i\in \calS_j^*[t-1]}h_i^{\prime}(x_i[t-1])+\frac{1-\xi}{|(N_j^-\cup \{j\})\cap \calN|-f}\sum_{i\in \calL_j^*[t-1]}h_i^{\prime}(x_i[t-1]).
\end{align*}

Define $q(x)$ as follows.
\begin{align}
\label{BS rewritten 1}
\nonumber
q(x)&=\frac{1}{|(N_j^-\cup \{j\})\cap \calN|-f}\sum_{k\in \calR_j^2[t-1]-\calF}\pth{\xi_k\, h_k(x)+(1-\xi_k)\, h_{j^{\prime}_k}(x)}\\
&\quad +\frac{\xi}{|(N_j^-\cup \{j\})\cap \calN|-f}\sum_{i\in \calS_j^*[t-1]}h_i(x)+\frac{1-\xi}{|(N_j^-\cup \{j\})\cap \calN|-f}\sum_{i\in \calL_j^*[t-1]}h_i(x).
\end{align}
In (\ref{BS rewritten 1}), for each $k\in \calR_j^2[t-1]-\calF$, it holds that
\begin{align*}
\frac{\xi_k}{|(N_j^-\cup \{j\})\cap \calN|-f}&\ge \frac{1}{2\pth{|(N_j^-\cup \{j\})\cap \calN|-f}}\\
&\ge \min\{\frac{1}{2\max_{i\in \calN} (d_i^-+1-\phi_i-f) }, \, \frac{1}{|\calN|}\}\\
&=\tilde{\beta}.
\end{align*}
 For each $i\in \calS_j^*[t-1]$, it holds that
\begin{align*}
\frac{\xi}{|(N_j^-\cup \{j\})\cap \calN|-f}&\ge \frac{1}{2\pth{|(N_j^-\cup \{j\})\cap \calN|-f}}\\
&=\tilde{\beta}.
\end{align*}

 In addition, we have
\begin{align*}
|\pth{\calR_j^2[t-1]-\calF}\cup \calS_j^*[t-1]|&=|\calR_j^2[t-1]-\calF|+|\calS_j^*[t-1]|\\
&=|\calR_j^2[t-1]|-|\calR_j^2[t-1]\cap \calF|+|\calS_j^*[t-1]|\\
&=d_j^-+1-2f-|\calR_j^2[t-1]\cap \calF|+f-\phi_j+|\calR_j^2[t-1]\cap \calF|\\
&=d_j^-+1-\phi_j-f=|(N_j^-\cup \{j\})\cap \calN|-f.
\end{align*}

Thus, in (\ref{BS rewritten 1}), at least $|(N_j^-\cup \{j\})\cap \calN|-f=d_j^-+1-\phi_j-f\ge \tilde{\gamma}$ non-faulty agents corresponding to agents $k\in \pth{\calR_j^2[t-1]-\calF}\cup \calS_j^*[t-1]$ are assigned with weights lower bounded by $\tilde{\beta}$. 

\paragraph{{\bf Case (ii)}: $\hat{g}_j[t-1]=\check{g}_j[t-1]$.}
Let $k\in \calR_j^2[t-1]-\calF$. Since $\hat{g}_j[t-1]\ge g_k[t-1]\ge \check{g}_j[t-1]$ and $\hat{g}_j[t-1]=\check{g}_j[t-1]$, it holds that $\hat{g}_j[t-1]=g_k[t-1]=\check{g}_j[t-1]$.
Consequently, we have
$$\widetilde{g}_j[t-1]=\frac{1}{2}\pth{\hat{g}_j[t-1]+\check{g}_j[t-1]}=g_k[t-1].$$
So we can rewrite $\widetilde{g}_j[t-1]$ as follows.
\begin{align*}
\nonumber
\widetilde{g}_j[t-1]&=\frac{|(N_j^-\cup \{j\})\cap \calN|-f}{|(N_j^-\cup \{j\})\cap \calN|-f}\,\widetilde{g}_j[t-1]\\
\nonumber
&=\frac{1}{|(N_j^-\cup \{j\})\cap \calN|-f}\pth{\sum_{k\in \calR_j^2[t-1]-\calF}\widetilde{g}_j[t-1]}+\frac{f-\phi_j+|\calR_j^2[t-1]\cap \calF|}{|(N_j^-\cup \{j\})\cap \calN|-f}\widetilde{g}_j[t-1]\\
\nonumber
&=\frac{1}{|(N_j^-\cup \{j\})\cap \calN|-f}\sum_{k\in \calR_j^2[t-1]-\calF}g_k[t-1]+\frac{\xi}{|(N_j^-\cup \{j\})\cap \calN|-f}\sum_{i\in \calS_j^*[t-1]}g_i[t-1]\\
\nonumber
&\quad+\frac{1-\xi}{|(N_j^-\cup \{j\})\cap \calN|-f}\sum_{i\in \calL_j^*[t-1]}g_i[t-1]\\
\nonumber
&=\frac{1}{|(N_j^-\cup \{j\})\cap \calN|-f}\sum_{k\in \calR_j^2[t-1]-\calF}h_k^{\prime}(x_k[t-1])\\
&\quad+\frac{\xi}{|(N_j^-\cup \{j\})\cap \calN|-f}\sum_{i\in \calS_j^*[t-1]}h_i^{\prime}(x_i[t-1])+\frac{1-\xi}{|(N_j^-\cup \{j\})\cap \calN|-f}\sum_{i\in \calL_j^*[t-1]}h_i^{\prime}(x_i[t-1]).
\end{align*}

Define $q(x)$ as follows.
\begin{align}
\label{BS rewritten 2}
\nonumber
q(x)&=\frac{1}{|(N_j^-\cup \{j\})\cap \calN|-f}\sum_{k\in \calR_j^2[t-1]-\calF}h_k(x)+\frac{\xi}{|(N_j^-\cup \{j\})\cap \calN|-f}\sum_{i\in \calS_j^*[t-1]}h_i(x)\\
&\quad+\frac{1-\xi}{|(N_j^-\cup \{j\})\cap \calN|-f}\sum_{i\in \calL_j^*[t-1]}h_i(x).
\end{align}
In (\ref{BS rewritten 2}), for each $k\in \calR_j^2[t-1]-\calF$, it holds that
\begin{align*}
\frac{1}{|(N_j^-\cup \{j\})\cap \calN|-f}&\ge \frac{1}{2\pth{|(N_j^-\cup \{j\})\cap \calN|-f}}=\frac{1}{2(d_j^-+1-\phi_j-f)}\\
&\ge \min\{\frac{1}{2\max_{i\in \calN} (d_i^-+1-\phi_i-f) }, \, \frac{1}{|\calN|}\}=\tilde{\beta}.
\end{align*}
 For each $i\in \calS_j^*[t-1]$, it holds that
\begin{align*}
\frac{\xi}{|(N_j^-\cup \{j\})\cap \calN|-f}&\ge \frac{1}{2\pth{|(N_j^-\cup \{j\})\cap \calN|-f}}=\frac{1}{2(d_j^-+1-\phi_j-f)}\\
&\ge \tilde{\beta}.
\end{align*}
In addition, we have
\begin{align*}
|\pth{\calR_j^2[t-1]-\calF}\cup \calS_j^*[t-1]|=|(N_j^-\cup \{j\})\cap \calN|-f.
\end{align*}
Thus, in (\ref{BS rewritten 2}), at least $|(N_j^-\cup \{j\})\cap \calN|-f=d_j^-+1-\phi_j-f\ge \tilde{\gamma}$ non-faulty agents corresponding to $\pth{\calR_j^2[t-1]-\calF}\cup \calS_j^*[t-1]$ are assigned with weights lower bounded by $\tilde{\beta}$.\\

Case (i) and Case (ii) together prove the lemma.

\eproof
\end{proof}

%

Recall that
$$\frac{1}{d_j^-+1-2f}\sum_{i\in \calR_j^1[t-1]\cup \{j\}} w_i[t-1]=\sum_{i\in \calR_j^1[t-1]\cup \{j\}}{\bf M}_{ji}[t-1] x_i[t-1],$$
where ${\bf M}_{ji}[t-1]$ is the entry of matrix ${\bf M}[t-1]$ at the $j$--th row and $i$--th column.

Let $\{z[t]\}_{t=0}^{\infty}$ be a sequence of estimates such that
\begin{align}
\label{alg3 crash sequence z}
z[t]=x_{j_{t}}[t], ~~\text{where}~j_{t}\in \argmax_{j\in \calN} Dist\pth{x_j[t], Y}.
\end{align}
From the definition, there is a sequence of agents $\{j_t\}_{t=0}^{\infty}$ associated with the sequence $\{z[t]\}_{t=0}^{\infty}$.

\begin{theorem}
\label{talgo BS}
Let $\tilde{\beta}=\min\{\frac{1}{2\max_{i\in \calN} (d_i^-+1-\phi_i-f) }, \, \frac{1}{|\calN|}\}$, and $\tilde{\gamma}=\min_{i\in \calN} (d_i^-+1-\phi_i-f)$.
The sequence $\{Dist\pth{z[t], Y(\tilde{\beta}, \tilde{\gamma})}\}_{t=0}^{\infty}$ converges and $$\lim_{t\diverge} Dist\pth{z[t], Y(\tilde{\beta}, \tilde{\gamma})}=0.$$
\end{theorem}
\begin{proof}

\begin{align}
\nonumber
Dist\pth{z[t+1], Y}&=Dist\pth{x_{j_{t+1}}[t], Y}~~~\text{by (\ref{alg3 crash sequence z})}\\
\nonumber
&=Dist\pth{\frac{1}{d_{j_{t+1}}^-+1-2f}\sum_{i\in \calR^{1}_{j_{t+1}}[t]\cup \{j_{t+1}\}} w_i[t]-\lambda[t]\widetilde{g}_{j_{t+1}}[t]+e_{j_{t+1}}[t],~ Y}~~~\text{by}~(\ref{Byzantine Iterative})\\
\nonumber
&=Dist\pth{\sum_{i\in (N_{j_{t+1}}^{-}\cup\{j_{t+1}\})\cap \calN} {\bf M}_{ji}[t] x_i[t]-\lambda[t]\widetilde{g}_{j_{t+1}}[t]+e_{j_{t+1}}[t],~ Y}\\
&=Dist\pth{\sum_{i\in (N_{j_{t+1}}^{-}\cup\{j_{t+1}\})\cap \calN} {\bf M}_{ji}[t] \pth{x_i[t]-\lambda[t]\widetilde{g}_{j_{t+1}}[t]}+e_{j_{t+1}}[t],~ Y}\label{convex}\\
\nonumber
&\le \sum_{i\in (N_{j_{t+1}}^{-}\cup\{j_{t+1}\})\cap \calN} {\bf M}_{ji}[t] \, Dist\pth{x_i[t]-\lambda[t]\widetilde{g}_{j_{t+1}}[t]+e_{j_{t+1}}[t], ~Y}~~~\text{by convexity of}~Dist\pth{\cdot, Y}\\
\nonumber
&\le \max_{i\in (N_{j_{t+1}}^{-}\cup\{j_{t+1}\})\cap \calN} Dist\pth{x_i[t]-\lambda[t]\widetilde{g}_{j_{t+1}}[t]+e_{j_{t+1}}[t], ~Y}.
\end{align}
Equality (\ref{convex}) holds due to the fact that $\sum_{i\in i\in (N_{j_{t+1}}^{-}\cup\{j_{t+1}\}} {\bf M}_{ji}=1$. By Lemma \ref{BS valid gradient}, there exists a valid function $p_{t}(\cdot)=\sum_{q\in \calN} \alpha_q h_q(\cdot)\in \calC$ such that
\begin{align}
\label{valid gradient inter}
\widetilde{g}_{j_{t+1}}[t]=\sum_{q\in \calN} \alpha_q h_q^{\prime}(x_q[t]).
\end{align}
In addition, let $$j_{t+1}^{\prime}\in \argmax_{i\in (N_{j_{t+1}}^{-}\cup\{j_{t+1}\})\cap \calN} Dist\pth{x_i[t]-\lambda[t]\widetilde{g}_{j_{t+1}}[t]+e_{j_{t+1}}[t], ~Y}.$$
We get
\begin{align}
\nonumber
Dist\pth{z[t+1], Y}
&\le \max_{i\in (N_{j_{t+1}}^{-}\cup\{j_{t+1}\})\cap \calN} Dist\pth{x_i[t]-\lambda[t]\widetilde{g}_{j_{t+1}}[t]+e_{j_{t+1}}[t], ~Y}~~~\text{by (\ref{BS a1})}\\
\nonumber
&= Dist\pth{x_{j_{t+1}^{\prime}}[t]-\lambda[t]\widetilde{g}_{j_{t+1}}[t]+e_{j_{t+1}}[t], ~Y}\\
\nonumber
&=Dist\pth{x_{j_{t+1}^{\prime}}[t]-\lambda[t]\sum_{q\in \calN} \alpha_q h_q^{\prime}(x_q[t])+e_{j_{t+1}}[t], ~Y}~~~\text{by Lemma \ref{BS valid gradient}}\\
\nonumber
&=\inf_{y\in Y}\left |x_{j_{t+1}^{\prime}}[t]-\lambda[t]\sum_{q\in \calN} \alpha_q h_q^{\prime}(x_q[t])-y +e_{j_{t+1}}[t] \right |\\
\nonumber
&=\inf_{y\in Y}\left |x_{j_{t+1}^{\prime}}[t]-\lambda[t]\sum_{q\in \calN} \alpha_q h_q^{\prime}(x_{j_{t+1}^{\prime}}[t])-y+ e_{j_{t+1}}[t]+ \lambda[t]\sum_{q\in \calN} \alpha_q \pth{h_q^{\prime}(x_{j_{t+1}^{\prime}}[t])-h_q^{\prime}(x_q[t])}\right |\\
\nonumber
&\le \inf_{y\in Y}\left |x_{j_{t+1}^{\prime}}[t]-\lambda[t]\sum_{q\in \calN} \alpha_q h_q^{\prime}(x_{j_{t+1}^{\prime}}[t])-y\right |+ |e_{j_{t+1}}[t]|+\lambda[t]\sum_{q\in \calN} \alpha_q \left |h_q^{\prime}(x_{j_{t+1}^{\prime}}[t])-h_q^{\prime}(x_q[t])\right |\\
&\le \inf_{y\in Y}\left |x_{j_{t+1}^{\prime}}[t]-\lambda[t]\sum_{q\in \calN} \alpha_q h_q^{\prime}(x_{j_{t+1}^{\prime}}[t])-y\right |+L\lambda[t]+\lambda[t]\sum_{q\in \calN} \alpha_q L\left |x_{j_{t+1}^{\prime}}[t]-x_q[t]\right | \label{error1}\\
&\le \inf_{y\in Y}\left |x_{j_{t+1}^{\prime}}[t]-\lambda[t]p_{t+1}^{\prime}(x_{j_{t+1}^{\prime}}[t])-y\right |+L\lambda[t]+\lambda[t]L (M[t]-m[t]), \label{BS distance y2}
\end{align}
where $p_{t}$ is defined in (\ref{valid gradient inter}).
Inequality (\ref{error1}) follows from Proposition \ref{proj e} and the fact that $h_q(\cdot)$ has $L$--Lipschitz gradient for each local function. Inequality (\ref{BS distance y2}) is true because
$$\left |x_{j_{t+1}^{\prime}}[t]-x_q[t]\right |\le \max_{i, j\in \calN} \pth{x_i[t]-x_j[t]}=\max_{i\in \calN} x_i[t]- \min_{j\in \calN}x_j[t]=M[t]-m[t].$$
Note that $p_{t+1}^{\prime}(x_{j_{t+1}^{\prime}}[t])$ is the gradient of a valid function at point $x_{j_{t+1}^{\prime}}[t]$ for $$\tilde{\beta}=\min\{\frac{1}{2\max_{i\in \calN} (d_i^-+1-\phi_i-f) }, \, \frac{1}{|\calN|}\}, \text{and}~ \tilde{\gamma}=\min_{i\in \calN} (d_i^-+1-\phi_i-f).$$\\

Recall that $\{z[t]\}_{t=0}^{\infty}$ is a sequence of estimates such that
$$z[t]=x_{j_{t}}[t], ~~\text{where}~j_{t}\in \argmax_{j\in \calN} Dist\pth{x_j[t], Y}.$$
Note that for each $t\ge 0$, there exists a non-faulty agent $j_{t}^{\prime}$ such that (\ref{BS distance y2}) holds, and there exists a sequence of agents $\{j_{t}^{\prime}\}_{t=0}^{\infty}$.
Let $\{x[t]\}_{t=0}^{\infty}$ be a sequence of estimates such that $x[t]=x_{j_{t+1}^{\prime}}[t]$.
Let $\{g[t]\}_{t=0}^{\infty}$ be a sequence of gradients such that $g[t]=p_{t}^{\prime}(x_{j_{t+1}^{\prime}}[t])$.\\

The remaining of the proof is identical to the proof of Theorem 2 in \cite{su2015fault}.

\eproof
\end{proof}

\section{Crash Fault Tolerance}

In this section, we first revisit the problem of reaching unconstrained iterative approximate consensus in the presence of crash failures, where only local communication and minimal memory carried across iterations are allowed.

\subsection{Iterative Approximate Crash Consensus}

We present a matrix representation of the states evolution of all the agents. This matrix representation allows us to derive a bound on the convergence rate which is independent of the timing of the occurrence of crash failures. With this matrix representation, we provide an alternative proof of the sufficiency of the topology condition found in \cite{charron2014approximate,tseng2015fault} under which the iterative approximate crash consensus is achievable.

\begin{definition}
\label{reduced crash}
For a given graph $G(\calV, \calE)$, a reduced graph $\calH_c$ under crash faults is a subgraph of $G(\calV, \calE)$ obtained by removing all the edges incident to up to $|\calF|$ nodes in $\calF$.
\end{definition}

Reaching consensus on highly dynamic networks is considered in \cite{charron2014approximate}, which incorporates our problem as a special case.  Since the network topological dynamic in our problem is caused by the crash failure, thus the dynamic has some ``monotone" structure over time.
Although unconstrained algorithms are considered in \cite{tseng2015fault}, where agents can exchange messages with agents via possible multi-hop communication,  it can be shown that the necessary condition in the next lemma is equivalent to the tight condition found in \cite{tseng2015fault}.
\begin{lemma}\cite{charron2014approximate,tseng2015fault}
\label{crash consensus nec}
Iterative approximate crash consensus can be achieved on a given graph $G(\calV, \calE)$ {\em only if} every reduced graph under crash faults of $G(\calV, \calE)$ contains a unique non-trivial weakly-connected component, and the subgraph induced by this weakly-connected component contains a source.
\end{lemma}

\begin{proof}
Suppose for the given graph $G(\calV, \calE)$, there exists a reduced subgraph $\calH_c$ that\\
(i) contains at least two non-trivial weakly-connected components, or\\
(ii) contains a unique non-trivial weakly-connected component, {\em but} the subgraph induced by this weakly-connected component does not contain a source.

Let $F\subseteq \calF$ be the set of nodes whose edges are removed. Let $G-F$ be the subgraph induced by $\calV-F$.
Consider the execution where all the agents in $F$ crash at the very beginning of this execution.

Both (i) and (ii) imply that there is no spanning trees in $G-F$. By \cite{ren2007information}, we know that consensus cannot be achieved on $G-F$. Thus, crash consensus cannot be achieved on $G(\calV, \calE)$, proving the lemma.


\eproof
\end{proof}

\begin{assumption}
\label{a2}
Every reduced graph (under crash faults) of $G(\calV, \calE)$ contains a unique non-trivial weakly-connected component, and the subgraph induced by this weakly-connected component contains a source.
%
\end{assumption}

In Lemma \ref{crash consensus nec}, we show that the Assumption \ref{a2} is necessary for iterative approximate crash consensus to be achievable. Next we show that Assumption \ref{a2} is also sufficient.  We prove this by construction.

Let $\calN[t]$ be the set of agents that have not crashed {\em by the beginning} of iteration $t$, and let $\bar{\calN}[t]$ be the set of agents that have not crashed {\em by the end} of iteration $t$. Note that $\bar{\calN}[t]\subseteq \calN[t]$,  $\calN[t]-\bar{\calN}[t]$ is the collection of agents that crash {\em during} iteration $t$, and that $\bar{\calN}[t]=\calN[t+1]$.
In particular, in iteration $t$, if an agent crashes after performing the update step in (\ref{crash e_Z11}), we say this agent crashes in iteration $t+1$. Recall that $x_i[0]$ be the initial state of agent $i\in \calV$.

\paragraph{}
\vspace*{8pt}\hrule

{\bf Algorithm 3 (crash consensus)}
\vspace*{4pt}\hrule

~

Steps to be performed by agent $i\in \calN[t]$ in the $t$-th iteration:
\begin{enumerate}

\item {\em Transmit step:} Transmit current state $x_i[t-1]$ on all outgoing edges.
\item {\em Receive step:} Receive values on incoming edges. These values form
multiset 
$\calR_i(t)$ of size at most $|N_i^{-}|$. 
Let $\ell_i(t)=|\calR_i(t)|$.

\item {\em Update step:}
%
Update its state as follows.
\begin{eqnarray}
x_i[t] ~ = ~\frac{1}{\ell_i(t)+1}\pth{\sum_{j\in \calR_i(t)\cup \{i\}} x_j[t-1]}. 
\label{crash e_Z11}
\end{eqnarray}
\end{enumerate}

~
\hrule

~

~

Since we are interested in unconstrained consensus, there is no projection involved in the update step.
Recall that, in iteration $t$, if an agent crashes after performing the update step in (\ref{crash e_Z11}), we say this agent crashes in iteration $t+1$.
With this convention, we know that only agents in $\bar{\calN}[t]$ (agents that have not crashed by the end of iteration $t$) will update their states according to (\ref{crash e_Z11}). Consequently, for each $i\notin \bar{\calN}[t]$, we assume that $x_i[t]=x_i[t-1]$.
%
%
Let ${\bf P}[t]\in \reals ^{n \times n}$ be a $n$ by $n$ matrix such that
for each $i\in \bar{\calN}[t]$,
\begin{align}
\label{update crash1}
{\bf P}_{ij}[t]=
\begin{cases}
~~\frac{1}{\ell_i(t)+1}, ~~~\text{if}~ j\in \calR_i(t)\cup \{i\}\\
~~0, ~~~\text{otherwise},
\end{cases}
\end{align}
and for each $i\notin \bar{\calN}[t]$,
\begin{align}
\label{update crash2}
{\bf P}_{ij}[t]=
\begin{cases}
~~1, ~~~\text{if}~ j=i\\
~~0, ~~~\text{otherwise}.
\end{cases}
\end{align}
Note that the update matrix ${\bf P}[t]$ is row-stochastic for each $t\ge 1$. At each iteration $t$, an agent can only receive messages from agents that have not crashed by the beginning of iteration $t$. Thus
${\bf P}_{ij}[t]=0$ for each $j\notin \calN[t]$, and
$$1=\sum_{j=1}^n {\bf P}_{ij}[t] = \sum_{j\in \calN[t]} {\bf P}_{ij}[t]+\sum_{j\notin \calN[t]} {\bf P}_{ij}[t]=\sum_{j\in \calN[t]} {\bf P}_{ij}[t]+0=\sum_{j\in \calN[t]} {\bf P}_{ij}[t].$$
Let ${\bf x}[t-1]\in \reals^{n}$ be a real vector of the local estimates at the beginning of iteration $t$ with ${\bf x}_j[t-1]=x_j[t-1]$ being the local estimate of agent $j\in \calV$ (can be faulty). 
\begin{align}
\label{matrix representation crash01}
{\bf x}[t]={\bf P}[t]{\bf x}[t-1].
\end{align}
Note that in (\ref{matrix representation crash01}), the the time index of the matrix is the same as the iteration time index. In contrast, in (\ref{update}) and (\ref{matrix representation}), the matrix time index is smaller than the iteration time index by one. The time index in the update matrix definition (\ref{update crash1}) and (\ref{update crash2}) is chosen for ease of notation.

Let $\calH_c $ be a subgraph of $G(\calV, \calE)$, with ${\bf H}$ as the adjacency matrix. In every iteration $t\ge 1$, and for every ${\bf P}[t]$, there exists a reduced graph $\calH_c[t]$ with adjacency matrix ${\bf H}[t]$ such that
\begin{align}
\label{iterative CS matrix lb}
{\bf P}[t]\ge \zeta\, {\bf H}[t],
\end{align}
where $\zeta=\frac{1}{d_{\max}\,+\,1}$ and $d_{\max}=\max_{j\in \calV}d_j^-$. Note that both $G(\calV, \calE)$ and $\calH_c[t]$ have the same vertex set $\calV$.
%
%
Equation (\ref{matrix representation crash01}) can be further expanded out as

\begin{align}
\label{MR evo crash 0}
\nonumber
{\bf x}[t]&=\pth{{\bf  P}[t]{\bf  P}[t-1]\ldots {\bf  P}[1]}{\bf x}[0]\\
&={\bf  \Psi}(t,1){\bf x}[0],
\end{align}
where ${\bf \Psi}(t, r)={\bf P}[t]{\bf P}[t-1]\ldots {\bf P}[r]$ is a backward product, and by convention, ${\bf \Psi}(t, t)={\bf P}[t]$ and ${\bf \Psi}(t, t+1)={\bf I}$. Since $\ones$ is a right eigenvector of ${\bf \Psi}(t, r)$ with eigenvalue 1, the product ${\bf \Psi}(t, r)$ is also a row-stochastic matrix. In addition, the product ${\bf \Psi}(t, r)$ has the following property.
\begin{proposition}
\label{p3}
Given $t\ge 1$, for $t^{\prime}\ge t$, $i\in \calN[t]$ and $j\notin \calN[t]$, it holds that
\begin{align*}
{\bf \Psi}_{ij}(t^{\prime}, t)=0.
\end{align*}
\end{proposition}

\begin{proof}
We prove this proposition by inducting on $t^{\prime}$. \\

When $t^{\prime}=t$, ${\bf \Psi}(t^{\prime}, t)={\bf P}[t]$. Since $j\notin \calN[t]$, no agents (rather than agent $j$ itself) can receive messages from
agent $j$, i.e., $j\notin \calR_k(t)\cup \{k\}$ for any $k\not=j$.
Since $i\in \calN[t]$ and $j\notin \calN[t]$, it holds that $i\not=j$.
By (\ref{update crash1}), we know ${\bf P}_{ij}[t]=0$.\\

Recall that ${\bf \Psi}(t^{\prime}+1, t)={\bf P}[t^{\prime}+1]{\bf \Psi}(t^{\prime}, t)$. We get
\begin{align*}
{\bf \Psi}_{ij}\pth{t^{\prime}+1, t}&=\sum_{k=1}^n {\bf P}_{ik}[t^{\prime}+1]{\bf \Psi}_{kj}(t^{\prime}, t)\\
&=\sum_{k\in \calN[t]} {\bf P}_{ik}[t^{\prime}+1]{\bf \Psi}_{kj}(t^{\prime}, t)+\sum_{k\notin \calN[t]} {\bf P}_{ik}[t^{\prime}+1]{\bf \Psi}_{kj}(t^{\prime}, t)\\
&=\sum_{k\in \calN[t]} {\bf P}_{ik}[t^{\prime}+1]\,0+\sum_{k\notin \calN[t]} {\bf P}_{ik}[t^{\prime}+1]{\bf \Psi}_{kj}(t^{\prime}, t)~~~\text{by induction hypothesis}\\
&=\sum_{k\notin \calN[t]} {\bf P}_{ik}[t^{\prime}+1]{\bf \Psi}_{kj}(t^{\prime}, t)\\
&\le\sum_{k\notin \calN[t^{\prime}+1]} {\bf P}_{ik}[t^{\prime}+1]{\bf \Psi}_{kj}(t^{\prime}, t)~~~\text{since } \calN[t^{\prime}+1]\subseteq \calN[t]\\ 
&=\sum_{k\notin \calN[t^{\prime}+1]} 0{\bf \Psi}_{kj}(t^{\prime}, t)~~~\text{since }{\bf P}_{ik}[t^{\prime}+1]=0, \forall\, k\notin \calN[t^{\prime}+1], i\not=k\\
&=0.
\end{align*}

In addition, ${\bf \Psi}_{ij}(t^{\prime}+1, t)\ge 0$. Thus,
$${\bf \Psi}_{ij}(t^{\prime}+1, t)= 0,$$
proving the induction.

\eproof
\end{proof}
As an immediate consequence of Proposition \ref{p3}, for any $t$, $t^{\prime}\ge t$, and $i\in \calN[t]$, the following holds.
\begin{align}
\label{sum}
\sum_{j\in \calN[t]} {\bf \Psi}(t^{\prime}+1, t)=1.
\end{align}

We next show that under Assumption \ref{a2}, the submatrix of ${\bf \Psi}(t,r)$ with rows and columns corresponding to all the non-faulty agents in $\calN$ will converge to a rank-one matrix of dimension $|\calN|\times |\calN|$.
\vskip 2\baselineskip

We first adapt the weak-ergodicity results obtained by Hajnal \cite{hajnal1958weak} for matrix products.

For each $t\ge 1$, $t^{\prime}\ge t$ and $r\ge 1$, we define $\delta_{r}({\bf \Psi}(t^{\prime}, t))$ and $\eta_{r}({\bf \Psi}(t^{\prime}, t))$ as follows.
\begin{align}
\delta_{r}({\bf \Psi}(t^{\prime}, t))&=\max_{\beta\in \calN[r]} \max_{\alpha, \alpha^{\prime} \in \calN[r]} \left | {\bf \Psi}_{\alpha \beta}(t^{\prime}, t)-{\bf \Psi}_{\alpha^{\prime} \beta}(t^{\prime}, t)\right |,\label{ergodic 1}\\
\eta_{r}({\bf \Psi}(t^{\prime}, t)) &= \min_{\alpha, \alpha^{\prime} \in \calN[r]} \sum_{\beta\in \calN[r]} \min \{{\bf \Psi}_{\alpha \beta}(t^{\prime}, t), {\bf \Psi}_{\alpha^{\prime} \beta}(t^{\prime}, t)\} \label{ergodic 2}.
\end{align}
Since $\calN[t]\subseteq \calN[r]$ for any $r\le t$, by the above definition, it is easy to see that
\begin{align}
\label{monotone 1}
\delta_{t}({\bf \Psi}(t^{\prime}, t))\le \delta_{r}({\bf \Psi}(t^{\prime}, t)),
\end{align}
and
\begin{align}
\label{monotone 2}
\eta_{t}({\bf \Psi}(t^{\prime}, t))\ge \eta_{r}({\bf \Psi}(t^{\prime}, t)).
\end{align}
Note that the definition of $\delta_r\pth{\bf \Psi}\pth{t^{\prime}, t}$ is symmetric in $\alpha$ and $\alpha^{\prime}$. Thus,
\begin{align}
\label{ergodic 11}
\nonumber
\delta_{r}({\bf \Psi}(t^{\prime}, t))&=\max_{\beta\in \calN[r]} \max_{\alpha, \alpha^{\prime} \in \calN[r]} \left | {\bf \Psi}_{\alpha \beta}(t^{\prime}, t)-{\bf \Psi}_{\alpha^{\prime} \beta}(t^{\prime}, t)\right |\\
&=\max_{\beta\in \calN[r]} \max_{\alpha, \alpha^{\prime} \in \calN[r]} \pth{ {\bf \Psi}_{\alpha \beta}(t^{\prime}, t)-{\bf \Psi}_{\alpha^{\prime} \beta}(t^{\prime}, t)}.
\end{align}

\begin{lemma}
\label{crash erdoc c1}
For each $t\ge 0$ and $t^{\prime}\ge t$, we have
\begin{align*}
\delta_{t}({\bf \Psi}(t^{\prime}, t))\le 1-\eta_{t}({\bf \Psi}(t^{\prime}, t)).
\end{align*}
\end{lemma}
\begin{proof}
\begin{align*}
1-\eta_{t}({\bf \Psi}(t^{\prime}, t))&=1-\min_{\alpha, \alpha^{\prime} \in \calN[t]} \sum_{\beta\in \calN[t]} \min \{{\bf \Psi}_{\alpha \beta}(t^{\prime}, t), {\bf \Psi}_{\alpha^{\prime} \beta}(t^{\prime}, t)\}~~~\text{by (\ref{ergodic 2})}\\
&=\max_{\alpha,\alpha^{\prime}\in \calN[t]}\pth{1-\sum_{\beta\in \calN[t]} \min \{{\bf \Psi}_{\alpha \beta}(t^{\prime}, t), {\bf \Psi}_{\alpha^{\prime} \beta}(t^{\prime}, t)\}}\\
&=\max_{\alpha,\alpha^{\prime}\in \calN[t]}\pth{\sum_{\beta\in \calN[t]} {\bf \Psi}_{\alpha \beta}(t^{\prime},t)-\sum_{\beta\in \calN[t]} \min \{{\bf \Psi}_{\alpha \beta}(t^{\prime}, t), {\bf \Psi}_{\alpha^{\prime} \beta}(t^{\prime}, t)\}}~~~\text{due to (\ref{sum})}\\
&=\max_{\alpha,\alpha^{\prime}\in \calN[t]}\sum_{\beta: ~\beta\in \calN[t], {\bf \Psi}_{\alpha \beta}(t^{\prime}, t) \ge {\bf \Psi}_{\alpha^{\prime} \beta}(t^{\prime}, t) } \pth{{\bf \Psi}_{\alpha \beta}(t^{\prime}, t)- {\bf \Psi}_{\alpha^{\prime} \beta}(t^{\prime}, t)}\\
&\ge \max_{\alpha,\alpha^{\prime}\in \calN[t]}\max_{\beta\in \calN[t]}\pth{ {\bf \Psi}_{\alpha \beta}(t^{\prime}, t)- {\bf \Psi}_{\alpha^{\prime} \beta}(t^{\prime}, t)}\\
&=\delta_t({\bf \Psi}(t^{\prime},t)) ~~~\text{by (\ref{ergodic 11})}
\end{align*}

\eproof
\end{proof}

\begin{lemma}
\label{crash erdoc c2}
 For $t_2> t_1\ge t_0\ge 1$,  define ${\bf P}={\bf \Psi}(t_2, t_1+1)$, ${\bf G}={\bf \Psi}(t_1, t_0)$, and ${\bf F}={\bf \Psi}(t_2, t_0)={\bf P}{\bf G}$.
Then the following is true
\begin{align*}
\delta_{t_1+1}({\bf F}) \le \pth{1-\eta_{t_1+1}({\bf P})}\delta_{t_1+1}({\bf G}).
\end{align*}
\end{lemma}

\begin{proof}
Since ${\bf P}={\bf \Psi}(t_2, t_1+1)$, ${\bf G}={\bf \Psi}(t_1, t_0)$, and ${\bf F}={\bf \Psi}(t_2, t_0)$, it holds that
$${\bf F}={\bf \Psi}(t_2, t_0)={\bf \Psi}(t_2, t_1+1){\bf \Psi}(t_1, t_0)={\bf P}{\bf G}.$$
For any $\alpha, \beta\in \calV$, we have ${\bf F}_{\alpha \beta}=\sum_{k=1}^n{\bf P}_{\alpha k}{\bf G}_{k \beta}$. We get
\begin{align}
\nonumber
\delta_{t_1+1}({\bf F})&= \max_{\beta \in \calN[t_1+1]} \max_{\alpha, \alpha^{\prime}\in \calN[t_1+1]} \pth{{\bf F}_{\alpha \beta}- {\bf F}_{\alpha^{\prime} \beta}}~~~\text{by (\ref{ergodic 11})}\\
\nonumber
&=\max_{\beta \in \calN[t_1+1]} \max_{\alpha, \alpha^{\prime}\in \calN[t_1+1]} \pth{\sum_{k=1}^n{\bf P}_{\alpha k}{\bf G}_{k \beta}- \sum_{k=1}^n{\bf P}_{\alpha^{\prime} k}{\bf G}_{k \beta}}\\
\nonumber
&=\max_{\beta \in \calN[t_1+1]} \max_{\alpha, \alpha^{\prime}\in \calN[t_1+1]} \pth{\sum_{k=1}^n\pth{{\bf P}_{\alpha k}- {\bf P}_{\alpha^{\prime} k}}{\bf G}_{k \beta}}\\
\nonumber
&=\max_{\beta \in \calN[t_1+1]} \max_{\alpha, \alpha^{\prime}\in \calN[t_1+1]} \pth{\sum_{k\in \calN[t_1+1]}\pth{{\bf P}_{\alpha k}- {\bf P}_{\alpha^{\prime} k}}{\bf G}_{k \beta}+ \sum_{k\notin \calN[t_1+1]}\pth{{\bf P}_{\alpha k}- {\bf P}_{\alpha^{\prime} k}}{\bf G}_{k \beta}}\\
&=\max_{\beta \in \calN[t_1+1]} \max_{\alpha, \alpha^{\prime}\in \calN[t_1+1]} \pth{\sum_{k\in \calN[t_1+1]}\pth{{\bf P}_{\alpha k}- {\bf P}_{\alpha^{\prime} k}}{\bf G}_{k \beta}+ \sum_{k\notin \calN[t_1+1]}\pth{{\bf P}_{\alpha k}- {\bf P}_{\alpha^{\prime} k}}{\bf G}_{k \beta}} \label{partition1}.
\end{align}
%
%
%
Recall that $\calN[t_1+1]$ is the collection of agents that have not crashed by the beginning of iteration $t_1+1$. If $k\notin \calN[t_1+1]$, then $k$ crashed in the first $t_1$ iterations. From (\ref{update crash2}), we know that if $k\notin \calN[t_1+1]$, then ${\bf P}_{ik}[t]=0$ for all $i\not=k$ and all $t>t_1$.
As ${\bf P}={\bf \Psi}(t_2, t_1+1)={\bf P}[t_2]{\bf P}[t_2-1]\cdots {\bf P}[t_1+1]$, we know that ${\bf P}_{i k}={\bf \Psi}_{i k}(t_2, t_1+1)=0$ for any $i\not=k$ and $k\notin \calN[t_1+1]$.
Intuitively speaking, for any $t>t_1$, since agent $k$ has already crashed, it cannot influence any other agents -- no other agents can receive any message from agent $k$ after iteration $t_1$.
 For $\alpha, \alpha^{\prime}\in \calN[t_1+1]$ and $k\notin \calN[t_1+1]$, then ${\bf P}_{\alpha k}=0={\bf P}_{\alpha^{\prime} k}$, i.e., for any $\alpha, \alpha^{\prime}\in \calN[t_1+1]$, it holds that
$$ \sum_{k\notin \calN[t_1+1]}\pth{{\bf P}_{\alpha k}- {\bf P}_{\alpha^{\prime} k}}{\bf G}_{k \beta}~=~\sum_{k\notin \calN[t_1+1]}0\, {\bf G}_{k \beta}~=~0.$$
Thus, equality (\ref{partition1}) can be simplified as
\begin{align*}
\delta_{t_1+1}({\bf F})~
&=\max_{\beta \in \calN[t_1+1]} \max_{\alpha, \alpha^{\prime}\in \calN[t_1+1]} \pth{\sum_{k\in \calN[t_1+1]}\pth{{\bf P}_{\alpha k}- {\bf P}_{\alpha^{\prime} k}}{\bf G}_{k \beta}}\\
&= \max_{\beta \in \calN[t_1+1]} \max_{\alpha, \alpha^{\prime}\in \calN[t_1+1]} \Bigg{(}\sum_{k:~k\in \calN[t_1+1], \,{\bf P}_{\alpha\,k}\ge {\bf P}_{\alpha^{\prime}\,k}}\pth{{\bf P}_{\alpha k}- {\bf P}_{\alpha^{\prime} k}}{\bf G}_{k \beta}\\
&\quad \quad \quad \quad \quad \quad \quad \quad \quad \quad \quad \quad+\sum_{k:~k\in \calN[t_1+1], \,{\bf P}_{\alpha\,k}< {\bf P}_{\alpha^{\prime}\,k}}\pth{{\bf P}_{\alpha k}- {\bf P}_{\alpha^{\prime} k}}{\bf G}_{k \beta} \Bigg{)}\\
&\le \max_{\beta \in \calN[t_1+1]} \max_{\alpha, \alpha^{\prime}\in \calN[t_1+1]} \Bigg{(}\sum_{k:~k\in \calN[t_1+1], \,{\bf P}_{\alpha\,k}\ge {\bf P}_{\alpha^{\prime}\,k}}\pth{{\bf P}_{\alpha k}- {\bf P}_{\alpha^{\prime} k}}\max_{k\in \calN[t_1+1]}{\bf G}_{k \beta}\\
&\quad \quad \quad \quad \quad \quad \quad \quad \quad \quad \quad \quad+\sum_{k:~k\in \calN[t_1+1], \,{\bf P}_{\alpha\,k}< {\bf P}_{\alpha^{\prime}\,k}}\pth{{\bf P}_{\alpha k}- {\bf P}_{\alpha^{\prime} k}}\min_{k\in \calN[t_1+1]}{\bf G}_{k \beta} \Bigg{)}\\
&= \max_{\alpha, \alpha^{\prime}\in \calN[t_1+1]} \pth{\sum_{k:~k\in \calN[t_1+1], \,{\bf P}_{\alpha\,k}\ge {\bf P}_{\alpha^{\prime}\,k}}\pth{{\bf P}_{\alpha k}- {\bf P}_{\alpha^{\prime} k}}}\max_{\beta \in \calN[t_1+1]} \pth{\max_{k\in \calN[t_1+1]}{\bf G}_{k \beta}-\min_{k\in \calN[t_1+1]}{\bf G}_{k \beta}}\\
\nonumber
&=\max_{\alpha, \alpha^{\prime}\in \calN[t_1+1]} \pth{\sum_{k:~k\in \calN[t_1+1], \,{\bf P}_{\alpha\,k}\ge {\bf P}_{\alpha^{\prime}\,k}}\pth{{\bf P}_{\alpha k}- {\bf P}_{\alpha^{\prime} k}}}\delta_{t_1+1}\pth{\bf G}\\
&=\max_{\alpha, \alpha^{\prime}\in \calN[t_1+1]} \sum_{k\in \calN[t_1+1]}\pth{{\bf P}_{\alpha\,k}-\min \{{\bf P}_{\alpha\,k}, {\bf P}_{\alpha^{\prime}\,k}\}}\delta_{t_1+1}({\bf G})\\
&=\max_{\alpha, \alpha^{\prime}\in \calN[t_1+1]} \pth{1-\sum_{k\in \calN[t_1+1]}\min \{{\bf P}_{\alpha\,k}, {\bf P}_{\alpha^{\prime}\,k}\}}\delta_{t_1+1}({\bf G})\\
&=\pth{1-\min_{\alpha, \alpha^{\prime}\in \calN[t_1+1]}\sum_{k\in \calN[t_1+1]}\min \{{\bf P}_{\alpha\,k}, {\bf P}_{\alpha^{\prime}\,k}\} }\delta_{t_1+1}({\bf G})\\
&=\pth{1-\eta_{t_1+1}({\bf P})}\delta_{t_1+1}({\bf G}),
\end{align*}
proving the lemma.

\eproof
\end{proof}

Let
\begin{align}
\label{block}
{\bf Q}[k]={\bf P}[kn] {\bf P}[kn-1]\ldots {\bf P}[(k-1)n+1].
\end{align}
Then
\begin{align}
\label{upper bound}
\nonumber
&\delta_{(k-1)n+1}\pth{{\bf \Psi}(kn,1)}=\delta_{(k-1)n+1}\pth{{\bf \Psi}(kn,(k-1)n+1){\bf \Psi}((k-1)n, 1)}\\
\nonumber
&\le \pth{1-\eta_{(k-1)n+1}\pth{{\bf \Psi}(kn,(k-1)n+1)}}\, \delta_{(k-1)n+1}\pth{{\bf \Psi}((k-1)n, 1)}~~~\text{by Lemma \ref{crash erdoc c2}}\\
\nonumber
&=\pth{1-\eta_{(k-1)n+1}\pth{{\bf Q}[k]}}\, \delta_{(k-1)n+1}\pth{{\bf \Psi}((k-1)n, 1)}~~~\text{by (\ref{block})}\\
\nonumber
&\le \pth{1-\eta_{(k-1)n+1}\pth{{\bf Q}[k]}}\, \delta_{(k-2)n+1}\pth{{\bf \Psi}((k-1)n, 1)}~~~\text{by (\ref{monotone 1})}\\
\nonumber
&\le \pth{1-\eta_{(k-1)n+1}\pth{{\bf Q}[k]}}\pth{1-\eta_{(k-2)n+1}\pth{{\bf Q}[k-1]}}\cdots \pth{1-\eta_{n+1}\pth{{\bf Q}[2]}}\delta_{1}\pth{{\bf \Psi}(n, 1)}~~\text{by Lemma \ref{crash erdoc c2}}\\
\nonumber
&\le \pth{1-\eta_{(k-1)n+1}\pth{{\bf Q}[k]}}\pth{1-\eta_{(k-2)n+1}\pth{{\bf Q}[k-1]}}\cdots \pth{1-\eta_{n+1}\pth{{\bf Q}[2]}} \pth{1-\eta_{1}\pth{{\bf Q}[1]}} ~~~\text{by Lemma \ref{crash erdoc c1}}\\
&= \prod_{r=1}^k \pth{1-\eta_{(r-1)n+1}\pth{{\bf Q}[r]}}.
\end{align}
Recall that for $t\ge 1$ and every ${\bf P}[t]$, there exists a reduced graph $\calH_c[t]$ with adjacency matrix ${\bf H}[t]$ such that
\begin{align*}
{\bf P}[t]\ge \zeta\, {\bf H}[t],
\end{align*}
where $\zeta=\frac{1}{d_{\max}\,+\,1}$.
Let $\widetilde{\calH}_c[k]=\cap_{r=(k-1)n+1}^{kn}\calH_c[r]$, whose vertex set is $\calV$, and edge set is the intersection of the edge sets of reduced graphs $\calH_c[r]$ for $r=(k-1)n+1, \cdots, kn$.
Intuitively speaking, $\widetilde{\calH}_c[k]$ is the maximal common subgraph of $\calH_c[r]$ for $r=(k-1)n+1, \ldots, kn$.
Note that $\widetilde{\calH}_c[k]$ is also a reduced graph. Let $\widetilde{{\bf H}}[k]$ be the adjacency matrix of $\widetilde{\calH}_c[k]$.
Since $\widetilde{\calH}_c[k]\subseteq \calH_c[r]$ for $r=(k-1)n+1, \ldots, kn$, it holds that
\begin{align}
\label{lb1}
{\bf H}[r]\ge \widetilde{{\bf H}}[k].
\end{align}
Then, we get
\begin{align}
\label{lb2}
\nonumber
{\bf Q}[k]&={\bf P}[kn] {\bf P}[kn-1]\ldots {\bf P}[(k-1)n+1]\\
\nonumber
&\ge \zeta^{n} \prod_{r=(k-1)n+1}^{kn} {\bf H}[r]~~~\text{by}~(\ref{iterative CS matrix lb})\\
&\ge \zeta^{n} \pth{\widetilde{\bf H}[k]}^{n}~~~\text{by}~(\ref{lb1}).
\end{align}
Since $G(\calV, \calE)$ satisfies Assumption \ref{a2} and $\widetilde{\calH}_c[k]$ is a reduced graph under crash fault of $G(\calV, \calE)$, there exists at least one node in $\calN[nk+1]$ that can reach every other nodes in $\calN[nk+1]$. Thus, in at least one column of of $\pth{\widetilde{\bf H}[k]}^{n}$ corresponding to agents in $\calN[nk+1]$, the $i$--th entry is lower bounded by 1, for each $i\in \calN[nk+1]$. Thus,
\begin{align}
\label{coeff}
\eta_{nk+1}\pth{{\bf Q}[k]}&=\min_{\alpha, \alpha^{\prime} \in \calN[nk+1]} \sum_{\beta\in \calN[nk+1]} \min \{{\bf Q}_{\alpha \beta}[k], {\bf Q}_{\alpha^{\prime} \beta}[k]\} \\
&\ge \zeta^n.
\end{align}
If $\calN[(k-1)n+1]=\calN[kn+1]$, i.e., no crash occurs from the start of iteration $(k-1)n+1$ to the end of iteration $kn$, then
\begin{align*}
\eta_{(k-1)n+1}\pth{{\bf Q}[k]}&=\min_{\alpha, \alpha^{\prime} \in \calN[(k-1)n+1]} \sum_{\beta\in \calN[(k-1)n+1]} \min \{{\bf Q}_{\alpha \beta}[k], {\bf Q}_{\alpha^{\prime} \beta}[k]\}\\
&=\min_{\alpha, \alpha^{\prime} \in \calN[kn+1]} \sum_{\beta\in \calN[kn+1]} \min \{{\bf Q}_{\alpha \beta}[k], {\bf Q}_{\alpha^{\prime} \beta}[k]\}\\
&=\eta_{kn+1}\pth{{\bf Q}[k]}\ge \zeta^n ~~~\text{by (\ref{coeff})}
\end{align*}
Let $k\ge f$. Since there are at most $f$ crash failures, it holds for at least $k-f$ indices among $r=1, \cdots, k$ satisfy the following
\begin{align*}
\eta_{(r-1)n+1}\pth{{\bf Q}[r]}\ge \zeta^n.
\end{align*}
In addition, from (\ref{ergodic 2}) we know that
$$\eta_{(r-1)n+1}\pth{{\bf Q}[r]}\ge 0$$
 for any $r=1, \cdots, k$.
Thus, (\ref{upper bound}) can be further bounded from above as
\begin{align}
\label{upper bound 2}
\delta_{(k-1)n+1}\pth{{\bf \Psi}(kn,1)}\le \prod_{r=1}^k \pth{1-\eta_{(r-1)n+1}\pth{{\bf Q}[r]}}\le (1-\zeta^n)^{k-f}.
\end{align}
Taking limit on both sides of (\ref{upper bound 2}) over $k$, we get
\begin{align*}
\lim_{k\diverge}\,\delta_{(k-1)n+1}\pth{{\bf \Psi}(kn,1)}\le \lim_{k\diverge}\,\pth{1-\zeta^n}^{k-f}~=~0.
\end{align*}

Let $t=kn+r$ for any $0\le r<n$. Define $k(t)=\lfloor \frac{t}{n}\rfloor$.
Thus,
\begin{align*}
\lim_{t \diverge}\,\delta_{t}\pth{{\bf \Psi}(t,1)}&=\lim_{t \diverge}\,\delta_{t}\pth{{\bf \Psi}(t,kn+1){\bf \Psi}(kn,1)}\\
&\le \lim_{t \diverge}\,\pth{1-\eta_{kn+1}({\bf \Psi}(t,kn+1))}\delta_{kn+1}\pth{{\bf \Psi}(kn,1)}~~~\text{by Lemma \ref{crash erdoc c2}}\\
&\le \lim_{k(t)\diverge}\,\delta_{kn+1}\pth{{\bf \Psi}(kn,1)}\\
&\le \lim_{k(t)\diverge}\, \delta_{(k-1)n+1}\pth{{\bf \Psi}(kn,1)}~~~\text{by (\ref{monotone 1})}\\
&\le  \lim_{k(t) \diverge}\,\delta_{(k-1)n}\pth{{\bf \Psi}(kn,1)}~=~0.
\end{align*}
That is
$$\lim_{t\diverge} \max_{\alpha,\alpha^{\prime}\in \calN[t]} \left | {\bf \Psi}_{\alpha \beta}(t,1)-{\bf \Psi}_{\alpha^{\prime} \beta}(t,1)\right |~=~0.$$
Thus, consensus is achieved, proving that Assumption \ref{a2} is also sufficient.

So far, we have proved the following theorem.
\begin{theorem}
\label{crash consensus}
Iterative approximate crash consensus can be achieved on a given graph $G(\calV, \calE)$ {\em if and only if} every reduced graph (under crash faults) of $G(\calV, \calE)$ contains a source component, i.e., $G(\calV, \calE)$ satisfies Assumption \ref{a2}.
\end{theorem}

\subsubsection{Properties of the Limiting Weights}
We first present a detailed characterization of the consensus value in terms of limiting weights of individual agents.

Let $[{\bf \Psi}(t,1)]_{\calN}$ be the $|\calN|$ by $|\calN|$ submatrix of ${\bf \Psi}(t,1)$ whose rows and columns correspond to agents in $\calN$. From the fact that $\lim_{t \diverge}\,\delta_{t}\pth{{\bf \Psi}(t,1)}=0$, we know that the submatrix $[{\bf \Psi}(t,1)]_{\calN}$ will converge to a rank one matrix, i.e.,
\begin{align*}
\lim_{t\diverge} [{\bf \Psi}(t,1)]_{\calN}=
\begin{bmatrix}
    \pi_{i_1}(1)       & \pi_{i_2}(1) & \pi_{i_3}(1) & \dots & \pi_{i_{|\calN|}}(1) \\
    \pi_{i_1}(1)       & \pi_{i_2}(1) & \pi_{i_3}(1) & \dots & \pi_{i_{|\calN|}}(1) \\
    \hdotsfor{5} \\
    \pi_{i_1}(1)       & \pi_{i_2}(1) & \pi_{i_3}(1) & \dots & \pi_{i_{|\calN|}}(1)
\end{bmatrix},
\end{align*}
where $i_j\in \calN$ and $[\pi_{i_1}(1)\, \pi_{i_2}(1) \, \pi_{i_3}(1) \, \cdots \, \pi_{i_{|\calN|}}(1) ]$ is a sub-vector of a row stochastic $\pi(1)$ such that $\sum_{j \in\calN} \pi_{j}(1)=1$,  and  $\pi_{j}(1)=0$ for $j\notin \calN$.
Indeed, the following holds for any $r$,
\begin{align*}
\lim_{t\diverge} [{\bf \Psi}(t,r)]_{\calN}=
\begin{bmatrix}
    \pi_{i_1}(r)       & \pi_{i_2}(r) & \pi_{i_3}(r) & \dots & \pi_{i_{|\calN|}}(r) \\
    \pi_{i_1}(r)       & \pi_{i_2}(r) & \pi_{i_3}(r) & \dots & \pi_{i_{|\calN|}}(r) \\
    \hdotsfor{5} \\
    \pi_{i_1}(r)       & \pi_{i_2}(r) & \pi_{i_3}(r) & \dots & \pi_{i_{|\calN|}}(r)
\end{bmatrix}.
\end{align*}
Note that the limiting row stochastic vector may depends on $r$. Let $\gamma$ be the minimal size of the source component of each reduced graph under crash faults.

\begin{lemma}
\label{lc1}
There are at least $\gamma$ columns in $[{\bf \Psi}(r+n-1,\, r)]_{\calN}$ that are lower bounded by $\zeta^n \ones$ component-wise for all $r$, where $\ones\in \reals^{n-\phi}$ is an all one vector of dimension $n-\phi$.
\end{lemma}
Lemma \ref{lc1} follows immediately from (\ref{lb1}) and (\ref{lb2}). Thus its proof is omitted.


\begin{proof}
For each $t=r, \ldots, r+n-1$, let $\calH_c[t]$ be a reduced graph with adjacency matrix ${\bf H}[t]$ such that
$${\bf P}[t]\ge \zeta  {\bf H}[t], $$
where $\zeta=\frac{1}{d_{\max}+1}$.

Let $\widetilde{\calH}$ be the reduced graph in which the edges incident to all the faulty nodes (nodes in $\calF$) are removed.
Let $\widetilde{\bf H}$ be the adjacency matrix of $\widetilde{\calH}$.
It is easy to see that $\widetilde{\calH}\subseteq \cap_{t=r}^{r+n-1} \calH_c[t]$.

We get
\begin{align*}
{\bf \Psi}(r+n-1,\, r)&={\bf P}[r+n-1]{\bf P}[r+n-2]\cdots{\bf P}[r]\\
&\ge \zeta^n \prod_{t=r}^{r+n-1}{\bf H}[t]\\
&\ge \zeta^n \pth{\widetilde{\bf H}}^n.
\end{align*}
Since $G(\calV, \calE)$ satisfies Assumption \ref{a2} and there are at least $\gamma$ nodes in the source component of $G-\calF$, there exist at least $\gamma$ nodes in $\calN$ that can reach every other node in $G-\calF$. Thus, at least $\gamma$ columns in $[\pth{\widetilde{\bf H}}^{n}]_{\calN}$ is lower bounded by 1, where $[\pth{\widetilde{\bf H}}^{n}]_{\calN}$ is the submatrix of $\pth{\widetilde{\bf H}}^{n}$ whose rows and columns correspond to agents in $\calN$. Since $|\calN|=n-\phi$, the lemma is proved.

\eproof
\end{proof}

\begin{lemma}
\label{lc2}
For any fixed $r$, at least $\gamma$ entries in $\pi(r)$ are lower bounded by $\zeta^n$, i.e., there exists $\calI_r\subseteq \calN$ such that $|\calI_r|\ge \gamma$ and
\begin{align*}
\pi_i(r)\ge \zeta^n,
\end{align*}
for each $i\in \calI_r$.
\end{lemma}
The proof of Lemma \ref{lc2} is identical to the proof of Lemma \ref{lblimiting}.


\begin{proof}
Given ${\bf \Psi}(r+n-1,r)$, let $\calI_r\subseteq \calN$ be the collection of columns that are lower bounded by $\zeta^n \ones$ in $[{\Psi}(r+n-1,r)]_{\calN}$.
By Lemma \ref{lc1}, we know that $|\calI_r|\ge \gamma$.

Since $\delta_{t\diverge} \pth{{\bf \Psi}(t,r)}=0$, and ${\bf \Psi}(t,r)$ is a row-stochastic matrix. Let ${\bf \Psi}_{i\cdot}(t,r)$ be the $i$--th row of ${\bf \Psi}(t,r)$.
For each $i\in \calN$, it holds that
$$\lim_{t\diverge} {\bf \Psi}_{i\cdot}(t,r)~=~\pi(r),$$
where $\sum_{j\in \calN} \pi_j(r)=1$ and $\pi_j(r)=0$ for each $j\notin \calN$.
Similarly, we have
for each $i\in \calN$, it holds that
$$\lim_{t\diverge} {\bf \Psi}_{i\cdot}(t,r+n)~=~\pi(r+n),$$
where $\sum_{j\in \calN} \pi_j(r+n)=1$ and $\pi_j(r+n)=0$ for each $j\notin \calN$.

In addition, let $i\in \calN$, we have
\begin{align*}
\pi_j(r)&=\lim_{t\diverge} {\bf \Psi}_{ij}(t, r)\\
&=\lim_{t\diverge} \pth{{\bf \Psi}(t, r+n){\bf \Psi}(r+n-1,r)}_{ij}\\
&=\lim_{t\diverge} \sum_{k=1}^n {\bf \Psi}_{kj}(t, r+n){\bf \Psi}_{ik}(r+n-1,r)\\
&=\lim_{t\diverge} \pth{\sum_{k\in \calN} {\bf \Psi}_{ik}(t, r+n){\bf \Psi}_{kj}(r+n-1,r)+\sum_{k\notin \calN} {\bf \Psi}_{ik}(t, r+n){\bf \Psi}_{kj}(r+n-1,r)}\\
&=\sum_{k\in \calN} \pth{\lim_{t\diverge}{\bf \Psi}_{ik}(t, r+n)}{\bf \Psi}_{kj}(r+n-1,r)+\sum_{k\notin \calN} \pth{\lim_{t\diverge}{\bf \Psi}_{ik}(t, r+n)}{\bf \Psi}_{kj}(r+n-1,r)\\
&=\sum_{k\in \calN} \pi_i(r+n){\bf \Psi}_{kj}(r+n-1,r)+\sum_{k\notin \calN} \pi_k(r+n){\bf \Psi}_{kj}(r+n-1,r)\\
&=\sum_{k\in \calN} \pi_k(r+n){\bf \Psi}_{kj}(r+n-1,r).
\end{align*}

For each $j\in \calI_r$, we have ${\bf \Psi}_{kj}(r+n-1,r)\ge \zeta^n$.
Then
$$\pi_j(r)=\sum_{k\in \calN} \pi_k(r+n){\bf \Psi}_{kj}(r+n-1,r)\ge \sum_{k\in \calN} \pi_k(r+n)\zeta^n =\zeta^n.$$
The last equality holds since $\sum_{k\in \calN} \pi_k(r+n)=1$.

In addition, since $|\calI_r|\ge \gamma$, then $\pi_j(r)\ge \zeta^n$ for at least $\gamma$ agents in $\calN$.

\eproof
\end{proof}

\subsubsection{Undirected Graphs}
When $G(\calV, \calE)$ is undirected, we will show that average consensus among all the non-faulty agents can be achieved, i.e.,  $\pi_i=\frac{1}{|\calN|}$ for each $i\in \calN$ is achievable.
We will modify the update step of Algorithm 3. With the new update step, the state evolution of the non-faulty agents can be represented by a sequence of doubly-stochastic matrices.

\paragraph{}
\vspace*{8pt}\hrule

{\bf Algorithm 4 (crash consensus)}
\vspace*{4pt}\hrule

~

Steps to be performed by agent $i\in \calN[t]$ in the $t$-th iteration:
\begin{enumerate}

\item {\em Transmit step:} Transmit current state $x_i[t-1]$ on all outgoing edges.
\item {\em Receive step:} Receive values on incoming edges. These values form
multiset 
$\calR_i(t)$ of size at most $|N_i^{-}|$.\footnote{Some agents in $N_i^-$ may have crashed.}

\item {\em Update step:}
%
Let $a_j[t-1]=\frac{1}{\max\left\{d_i^-+1,\, d_j^-+1\right\}}$ for each $j\in \calR_i(t)$, and let $a_i[t-1]=1-\sum_{j\in \calR_i(t)} a_j[t-1]$.
Update its state as follows.
\begin{eqnarray}
x_i[t] ~ = ~\pth{\sum_{j\in \calR_i(t)}a_j[t-1]x_j[t-1]}+ a_i[t-1]x_i[t-1].
\label{crash e}
\end{eqnarray}
\end{enumerate}

~
\hrule

~

~
Note that to execute Algorithm 4, each agent needs to know the degrees of its neighbors.

Let $\widetilde{\bf P}[t]\in \reals^{n \times n}$ by a $n$ by $n$ matrix such that for each $i\in \bar{\calN}[t]$,
\begin{align}
\label{update crash3}
\widetilde{\bf P}_{ij}[t]=
\begin{cases}
~~\frac{1}{\max\{d_i^-+1,\, d_j^-+1\}}, ~~~\text{if}~ j\in \calR_i(t),\\
~~1-\sum_{k\in \calR_i(t)} \frac{1}{\max\{d_i^-+1,\, d_k^-+1\}} ~~~\text{if}~ j=i,\\
~~0, ~~~\text{otherwise},
\end{cases}
\end{align}
and for each $i\notin \bar{\calN}[t]$,
\begin{align}
\label{update crash4}
\widetilde{\bf P}_{ij}[t]=
\begin{cases}
~~\frac{1}{\max\{d_i^-+1,\, d_j^-+1\}}~~~\text{if}~ i\in \calR_j(t)\\
~~1-\sum_{k: j\in \calR_k(t)}\frac{1}{\max\{d_i^-+1,\, d_k^-+1\}}, ~~~\text{if}~ j=i\\
~~0, ~~~\text{otherwise}.
\end{cases}
\end{align}
It is easy to see that $\widetilde{\bf P}[t]$ is doubly-stochastic for all $t\ge 0$.

Let $\widetilde{\bf x}[t]\in \reals^{n}$ be a real vector of dimension $n$, with $\tx_j[t]$ defined as
$$ \tx_j[0]=x_j[0],$$
and
$$\tx_j[t]~=~\sum_{k=1}^n  \widetilde{\bf P}_{jk}[t] \tx_k[t-1],$$
for each $t\ge 1$, and for all $i\in \calV$. The evolution of vector $\widetilde{\bf x}[t]$ can be represented as follows.
\begin{align}
\label{eff evo1}
\widetilde{\bf x}[t]=\widetilde{\bf P}[t]\,\widetilde{\bf x}[t-1]=\widetilde{\bf P}[t]\widetilde{\bf P}[t-1]\cdots\widetilde{\bf P}[1]\,\widetilde{\bf x}[0].
\end{align}

\begin{proposition}
\label{eff evo}
For each $t\ge 0$, it holds that
\begin{align*}
\tx_j[t]=x_j[t],
\end{align*}
for each $j\in \bar{\calN}[t]$.
\end{proposition}
\begin{proof}
We prove this proposition by induction.

For $t=0$, as $\bar{\calN}[0]=\calV$, it follows trivially from the definition of $\tx_j[0]$ that $\tx_j[0]=x_j[0]$ for $j\in \calV$.

Suppose the proposition holds for the $t$--th iteration. Now consider the $t+1$--iteration.
\begin{align}
\nonumber
&\tx_j[t+1]=\sum_{k=1}^n \widetilde{\bf P}_{jk}[t+1] \tx_k[t]\\
\nonumber
&=\sum_{k\in \calR_j(t+1)} \frac{1}{\max\{d_j^-+1, \, d_k^-+1\}}\tx_k[t]+\pth{1-\sum_{k\in \calR_j(t+1)} \frac{1}{\max\{d_j^-+1,\, d_k^-+1\}} }\tx_j[t]\\
&=\sum_{k\in \calR_j(t+1)} \frac{1}{\max\{d_j^-+1, \, d_k^-+1\}}x_k[t]+\pth{1-\sum_{k\in \calR_j(t+1)} \frac{1}{\max\{d_j^-+1,\, d_k^-+1\}} }x_j[t]~\label{c1}\\
\nonumber
&=x_j[t+1]~~~\text{by~(\ref{crash e})}
\end{align}
Since $\calR_j(t+1)\cup \{j\}\subseteq \calN[t+1]=\bar{\calN}[t]$, then by the induction hypothesis, we know that
$\tx_i[t]=x_i[t]$ for each $i\in \calR_j(t+1)\cup \{j\}$. Therefore, the proof of the proposition is complete.

\eproof
\end{proof}

From Proposition \ref{eff evo}, we know that the evolution of $\widetilde{\bf x}[t]$ in (\ref{eff evo1}) captures the dynamic of the evolution of $x_j[t]$ for each $i\in \calN[t]$. With the same proof as Algorithm 3, we can show that
$$\lim_{t\diverge}\, \delta_{t}\pth{\widetilde{\bf P}[t] \widetilde{\bf P}[t-1]\cdots \widetilde{\bf P}[1]}~=~0.$$
The submatrix $[\widetilde{\bf P}[t] \widetilde{\bf P}[t-1]\cdots \widetilde{\bf P}[1]]_{\calN}$ will converge to a rank one matrix. That is, the rows in $\widetilde{\bf P}[t] \widetilde{\bf P}[t-1]\cdots \widetilde{\bf P}[1]$ corresponding to all the non-faulty agents in $\calN$ will be identical asymptotically. Let $\tilde{\pi}(1)$ be the identical row stochastic vector. It can be shown that $\sum_{i\in \calN}\pi_i(1)=1$ and $\pi_i(1)=0$ for each $i\notin \calN$. In addition, since $\widetilde{\bf P}[t]$ is a doubly-stochastic matrix, it holds that $\tilde{\pi}_i(1)=\frac{1}{|\calN|}$ for each $i\in \calN$.

\subsection{Crash-Tolerant Algorithms}

  We are interested in the following family of functions.
\begin{definition}
\label{validfCrash}
Let $\widetilde{\calA}(\beta, \gamma)$ be the collection of functions defined as follows:
\begin{align}
\nonumber
\widetilde{\calA}(\beta, \gamma)~\triangleq ~\Bigg{\{}~p(x): p(x)&=\sum_{i\in \calV} \alpha_i h_i(x), ~\forall i\in\calV, ~ \alpha_i\geq 0,\\
&\sum_{i\in \calV}\alpha_i=1,\text{~~and~~}
\pth{\sum_{i\in\calN} {\bf 1}\left\{\alpha_i\ge \beta\right\}} ~ \geq ~ \gamma ~\Bigg{\}}
\label{Cvalid collection}
\end{align}

\end{definition}

In contrast to Definition \ref{validfByzantine}, each function in $\widetilde{\calA}(\beta, \gamma)$ is a convex combination of {\em all} the local functions.
Define the collection of unconstrained optimal solutions as
$$\widetilde{X}(\beta, \gamma)~\triangleq ~ \cup_{p(x)\in \widetilde{\calA}(\beta, \gamma)} \, \arg\min_{x\in \reals} \, p(x).$$

\begin{lemma}
\label{convex C1}
If $\beta\le \frac{1}{|\calN|}$ and $\gamma\le |\calN|$, the set $\widetilde{X}(\beta, \gamma)$ is convex and closed.
\end{lemma}
The proof of Lemma \ref{convex C1} is similar to the proof of Lemma 10 and Lemma 11 in \cite{su2015fault}.

However, points in $\widetilde{X}(\beta, \gamma)$ may be infeasible. Thus, we define
$$\widetilde{Y}(\beta, \gamma)~\triangleq ~ \cup_{p(x)\in \widetilde{\calA}(\beta, \gamma)} \, \arg\min_{x\in \calX} \, p(x).$$

\begin{lemma}
\label{constant c}
If $\beta\le \frac{1}{|\calN|}$ and $\gamma\le |\calN|$, the set $\widetilde{Y}(\beta, \gamma)$ is convex and closed.
\end{lemma}
The proof of Lemma \ref{constant c} is the same as Lemma \ref{convex B cons}.

%

%

\paragraph{}
\vspace*{8pt}\hrule

{\bf Algorithm 5 (under crash faults)}
\vspace*{4pt}\hrule

~

Steps to be performed by agent $i\in \calN[t]$ in the $t$-th iteration:
\begin{enumerate}

\item {\em Transmit step:} Transmit current state $x_i[t-1]$ on all outgoing edges.
\item {\em Receive step:} Receive values on incoming edges. These values form
multiset 
$\calR_i(t)$ of size at most $|N_i^{-}|$.\footnote{Some agents in $N_i^-$ may have crashed.}
Let $\ell_i(t)=|\calR_i(t)|$.

\item {\em Update step:}
%
Compute $h_i^{\prime}\pth{x_i[t-1]}$ -- the gradient of agent $i$'s objective function $h_i(\cdot)$ at $x_i[t-1]$.

Update its state as follows.
\begin{eqnarray}
x_i[t] ~ = P_X\left [~\frac{1}{\ell_i(t)+1}\pth{\sum_{j\in \calR_i(t)\cup \{i\}} x_j[t-1]} \,-\lambda[t-1]~ h_i^{\prime}(x_i[t-1])\right ].
\label{crash e_Z}
\end{eqnarray}
\end{enumerate}

~
\hrule

~

~
\begin{align}
v_i[t-1]&=\frac{1}{\ell_i(t)+1}\pth{\sum_{j\in \calR_i(t)\cup \{i\}} x_j[t-1]} ,\label{proj31}\\
e_i[t-1]&=P_X\left [v_i[t-1]-\lambda[t-1]~ h_i^{\prime}(x_i[t-1])\right ]-\pth{v_i[t-1]-\lambda[t-1]~ h_i^{\prime}(x_i[t-1])} \label{proj32}.
\end{align}

\begin{theorem}
\label{CT}
For a given graph $G(\calV, \calE)$ and $\beta \le  \frac{1}{(d_{\max}+1)^{n}}$, if each reduced graph (under crash faults) $\calH_c$ contains a source component with size at least $\gamma$, where $\gamma\ge 1$, then Algorithm 5 optimizes a function in $\widetilde{\calA}(\beta, \gamma)$.

\end{theorem}
 Recall that $d_{\max}=\max_{j\in \calV} d_j^-$. 
Let ${\bf d}[t]\in \reals^{n}$ be a real vector of with $d_i[t]$ being the gradient of function $h_i(\cdot)$ at $x_i[t]$, and let ${\bf e}[t]\in \reals^{n}$ be the vector of projection errors defined in (\ref{proj32}). In particular, $d_i[t]=0$ if $i\notin \bar{\calN}[t]$. Then we have
\begin{align*}
{\bf x}[t]&={\bf P}[t]{\bf x}[t-1]-\lambda[t-1]{\bf d}[t-1]+{\bf e}[t-1]\\
&={\bf \Psi} (t, 1){\bf x}[0]-\sum_{r=1}^{t}\lambda[r-1]{\bf \Psi} (t, r+1){\bf d}[r-1]+\sum_{r=1}^{t}{\bf \Psi} (t, r+1){\bf e}[r-1].
\end{align*}

Suppose that all agents (both non-faulty agents and faulty agents) cease computing $h_i^{\prime}(x_i[t])$ after some time $\bar{t}$, i.e., after $\bar{t}$ gradient is replaced by 0.

Let $\{\bar{\bf x}[t]\}$ be the sequences of local estimates generated in this case. We have
\begin{align*}
\bar{\bf x}[\bar{t}+s]&={\bf \Psi} (\bar{t}+s, 1){\bf x}[0]-\sum_{r=1}^{\bar{t}}\lambda[r-1]{\bf \Psi} (\bar{t}+s, r+1){\bf d}[r-1]+\sum_{r=1}^{\bar{t}}{\bf \Psi} (\bar{t}+s, r+1){\bf e}[r-1].
\end{align*}

For each $i\in \calV$, we have
\begin{align*}
\bar{\bf x}_i[\bar{t}+s]&=\sum_{k=1}^n {\bf \Psi}_{ik} (\bar{t}+s, 1)x_k[0]
-\sum_{r=1}^{\bar{t}}\lambda[r-1]\sum_{k=1}^n{\bf \Psi}_{ik} (\bar{t}+s, r+1)d_k[r-1]\\
&\quad +\sum_{r=1}^{\bar{t}}{\bf \Psi}_{ik} (\bar{t}+s, r+1)e_k[r-1].
\end{align*}
For each $i\in \calN$, as $s\diverge$, we get
\begin{align*}
\lim_{s\diverge}\, \bar{\bf x}_i[\bar{t}+s]&=\lim_{s\diverge}\,\sum_{k=1}^n {\bf \Psi}_{ik} (1)x_k[0]
-\sum_{r=1}^{\bar{t}}\lambda[r-1]\lim_{s\diverge}\,\sum_{k=1}^n{\bf \Psi}_{ik} (r+1)d_k[r-1]\\
&\quad +\sum_{r=1}^{\bar{t}}\lim_{s\diverge}{\bf \Psi}_{ik} (\bar{t}+s, r+1)e_k[r-1]\\
&=\sum_{k=1}^n \pi_k(1)x_k[0]-\sum_{r=1}^{\bar{t}}\lambda[r-1]\sum_{k=1}^n\pi_k(r+1)d_k[r-1]+\sum_{r=1}^{\bar{t}}\sum_{k=1}^n\pi_k(r+1)e_k[r-1]
\end{align*}
Let $y[\bar{t}]=\lim_{s\diverge}\, \bar{\bf x}_i[\bar{t}+s]$.
If, instead, all agents cease computing $h_i^{\prime}(x_i[t])$ after iteration $\bar{t}+1$,  then the identical value, denoted by $y[(\bar{t}+1)]$, equals
\begin{align}
\nonumber
y[(\bar{t}+1)]&=\sum_{k=1}^n \pi_k(1)x_k[0]-\sum_{r=1}^{\bar{t}+1}\lambda[r-1]\sum_{k=1}^n\pi_k(r+1)d_k[r-1]+\sum_{r=1}^{\bar{t}+1}\sum_{k=1}^n\pi_k(r+1)e_k[r-1]\\
&=y[\bar{t}]-\lambda[\bar{t}]\sum_{k=1}^n \pi_k(\bar{t}+2)d_k[\bar{t}]+\sum_{k=1}^n\pi_k(\bar{t}+2)e_k[\bar{t}].
\label{ydynamic crash}
\end{align}

The remaining proof of Theorem \ref{CT} is identical to the proof of Theorem \ref{BT}.

\begin{remark}
If the underlying graph is undirected, i.e., $G(\calV, \calE)$ is undirected, we can modify the update step in Algorithm 5 as follows: \\
Let $a_j[t-1]=\frac{1}{\max\left\{d_i^-+1,\, d_j^-+1\right\}}$ for each $j\in \calR_i(t)$, and let $a_i[t-1]=1-\sum_{j\in \calR_i(t)} a_j[t-1]$.\\
Update its state as follows.
\begin{eqnarray*}
x_i[t] ~ = ~P_X\left [\sum_{j\in \calR_i(t)}a_j[t-1]x_j[t-1]+ a_ix_i[t-1]-\lambda[t-1] h_i^{\prime}(x_i[t-1])\right ].
\end{eqnarray*}
Then we are able to show the modified algorithm optimizes the function
$$\frac{1}{|\calN|} \sum_{i\in \calN}h_i(x).$$

\end{remark}

\vskip 2\baselineskip

For the case when $G(\calV, \calE)$ is directed, we have the following alternative algorithm.
\paragraph{}
\vspace*{8pt}\hrule
~

{\bf Algorithm 6 (under crash faults)} for agent $j$ for iteration $t\ge 0$:
~
\vspace*{4pt}\hrule

\begin{list}{}{}
\item[{\bf Step 1:}]
Compute $h_j^{\prime}\pth{x_j[t-1]}$ -- the gradient of the local cost function $h_j(\cdot)$ at point $x_j[t-1]$, and send the estimate and gradient pair $(x_j[t-1], h_j^{\prime}\pth{x_j[t-1]})$ to on all outgoing edges. \\

~
\item[{\bf Step 2:}]
Let $\calR_j[t-1]$ denote the multi-set of tuples of the form $\pth{x_i[t-1], \, h_i^{\prime}(x_i[t-1])}$ received on all incoming edges 
as a result of step 1.\\
Note that the size of $\calR_j[t-1]$ at most $|N_i^{-}|$.\footnote{Some agents in $N_i^-$ may have crashed.}
Let $\ell_i(t)=|\calR_j[t-1]|$.



~

\item[{\bf Step 3:}] 
%
%
%
Update its state as follows.
\begin{align*}
x_j[t]=P_X\left [\frac{1}{\ell_j+1} \pth{\sum_{i\in \calR_j^1[t-1]\cup \{j\}}x_i[t-1]-\lambda[t-1] h_i^{\prime}(x_i[t-1])}\right ]
\end{align*}

\end{list}

\hrule

~

~
Recall that $\phi_i=|\calN_i^-\cap \calF|$ for each $i\in \calV$.
\begin{theorem}
\label{talgo CS}
Let $\tilde{\beta}=\min\{\frac{1}{\max_{i\in \calV} (d_i^-+1-\phi_i) }, \, \frac{1}{|\calN|}\}$, and $\tilde{\gamma}=\min_{i\in \calV} (d_i^-+1-\phi_i)$,
Algorithm 6 optimizes a function in $\widetilde{\calA}(\beta, \gamma)$.
\end{theorem}

The proof of Theorem \ref{talgo CS} is similar to the proof of Theorem \ref{talgo BS}.

\section{Discussion}
We study the problem of constrained distributed optimization in multi-agent networks when some of the computing agents may be faulty. In this problem, the system goal is to have all the non-faulty agents collectively minimize a global objective given by weighted average of local cost functions, each of which is initially known to a non-faulty agent only. We focus on the family of algorithms considered in \cite{su2015fault}, where only local communication and minimal memory carried across iterations are allowed.
We generalize our previous results on fully-connected networks and unconstrained optimization  \cite{su2015fault}  to arbitrary directed networks and constrained optimization.  As a byproduct, we provide a matrix representation for iterative approximate crash consensus. The matrix representation allows us to characterize the convergence rate for crash iterative crash consensus.

In terms of solvable $\gamma$, Algorithm 1 (Algorithm 5) relies on the size of the source components of individual reduced graphs. In contrast, the solvable $\gamma$ by Algorithm 2 (Algorithm 6) only relies on the incoming degree of each agent. We believe that there will be a way to combine the analysis for Algorithm 1 and Algorithm 2 to get a better bound on $\gamma$. 
%
We leave this problem for further exploration.

%
%
%
%
%
%
%



\bibliographystyle{abbrv}
\bibliography{PSDA_DL}

\end{document}